\documentclass[11pt,a4paper]{article}
\usepackage[margin=1in]{geometry}
\usepackage{mathtools}
\usepackage{amsmath}
\usepackage{amssymb}
\usepackage{amsthm}
\usepackage[hidelinks]{hyperref}
\usepackage{mathrsfs}
\usepackage{amsbsy}
\usepackage{upgreek}
\usepackage[table]{xcolor}
\usepackage{enumerate}
\usepackage{todonotes}
\usepackage[normalem]{ulem}
\usepackage{graphicx}
\usepackage{caption}
\usepackage{subcaption}
\def\A{{\cal A}}

\def\OO{{\cal O}} 

\usepackage{authblk}
\usepackage{mathtools}

\DeclarePairedDelimiter\abs{\lvert}{\rvert}%
\DeclarePairedDelimiter\norm{\lVert}{\rVert}%

\makeatletter
\let\oldabs\abs
\def\abs{\@ifstar{\oldabs}{\oldabs*}}
\let\oldnorm\norm
\def\norm{\@ifstar{\oldnorm}{\oldnorm*}}

 \newtheorem{claim}{Claim}
 \newtheorem{remark}{Remark}
 \newtheorem{definition}{Definition}
\newtheorem{theorem}{Theorem}
 \newtheorem{lemma}{Lemma}
 \newtheorem{problem}{Problem}
 
\newtheorem{observation}{Observation}

\newcommand{\IR}{\mathbb{R}}

\usepackage{color}
\definecolor{Gray}{gray}{0.15}

\newcommand*{\email}[1]{\href{mailto:#1}{\nolinkurl{#1}} }

\title{Online  Dominating Set and Independent Set}
\author{ Minati De}  \author{Sambhav Khurana} \author{Satyam Singh}
\affil{Dept. of Mathematics,  
Indian Institute of Technology Delhi, New Delhi,
India-110016\\
\email{{minati,satyam.singh}@maths.iitd.ac.in, sambhav.khurana@alumni.iitd.ac.in}
}

\begin{document}


\date{\today}

\maketitle

\begin{abstract}



Finding minimum dominating set and maximum  independent set for graphs in the classical online setup are notorious  due to their disastrous $\Omega(n)$ lower bound of the competitive ratio that  even holds for interval graphs, where $n$ is the number of vertices. In this paper, inspired by Newton number, first, we introduce the independent kissing number~$\zeta$ of a graph.    We prove that the well known online greedy algorithm for dominating set achieves optimal competitive ratio~$\zeta$ for any graph.  We show that the same greedy algorithm achieves optimal competitive ratio~$\zeta$ for  online maximum independent set of a class of graphs with independent kissing number $\zeta$.  For minimum connected dominating set problem, we prove that  online greedy algorithm achieves an asymptotic  competitive ratio of~$2(\zeta-1)$, whereas for a family of translated convex objects the  lower bound is~$\frac{2\zeta-1}{3}$.  Finally, we study the value of $\zeta$ for some specific families of geometric objects: fixed and arbitrary oriented unit hyper-cubes in $\IR^d$,  congruent balls in $\IR^3$, fixed oriented unit triangles, fixed and arbitrary oriented regular polygons in $\IR^2$. For each of these families, we also present lower bounds of the minimum connected dominating set problem.



\noindent
\textbf{Keywords:}  independent kissing number, competitive ratio, dominating set, geometric intersection graph,   independent set, online algorithm.
\end{abstract}

\section{Introduction}
In this paper, we have considered online versions of some well known NP-hard problems on graphs. For a graph $G=(V,E)$,  a subset $D$ of $V$ is 
a \emph{{dominating set}} (DS)   if each vertex $v \in V$ is either in $D$ or it is adjacent to some vertex in $D$. The \emph{{minimum dominating set}} (MDS) problem involves finding  a dominating set of $G$ that has minimum number of vertices. A dominating set $D$ is said to be a \emph{{connected dominating set}} (CDS) if  the induced subgraph $G[D]$ is connected (if $G$ is not connected then $G[D]$ must be connected for each connected component of $G$). A dominating set $D$ is  said to be an \emph{{independent dominating set}} (IDS) if the induced subgraph $G[D]$ is an independent set. Note that an \emph{{independent set}} (IS) of a graph $G(V,E)$ is a subset $I\subseteq V$ such that no two  vertices of $I$ are adjacent to each other.  The \emph{{maximum independent set}} (MIS) problem  is to find an independent set with maximum number of vertices.

Apart from theoretical implications, there are several  applications of  these problems. For example, a DS can be thought of as transmitting stations that can transmit messages to all stations in the network~\cite{liu9introduction}. A CDS is  often used as the virtual backbone in a network where nodes in the CDS are responsible to relay messages~\cite{eriksson1994mbone, GuhaK98}.
A relevant issue arises regarding the maintenance of the DS or CDS as the
network topology changes with time. On the other hand, the IS plays a vital role in  resource scheduling~\cite{HalldorssonIMT02}. Requests for resources or sets of resources arrive online, and two requests can only be serviced simultaneously if they do not involve the same resource.

Online computation models  the real world phenomena: irreversibility of time (or decisions) and uncertainty of future (or  order of input). As a result, online models of computation get a considerable amount of attention in recent time. In the classical online setting for graph problems, the graph is revealed vertex by vertex, together with edges to all previously revealed vertices. One needs to maintain a feasible solution for the revealed vertices, and need to  make a decision whether to include the newly revealed vertex into the solution set.  The decision to include a vertex in the solution set can not be changed  later. Unfortunately,  in this classical online model, for  both MDS and MIS problems, the competitive ratio has a lower bound of $\Omega(n)$ that holds even for interval graphs, where $n$ is the number of vertices revealed to the algorithm~\cite{KingT97,LiptonT94}. Due to this hopeless lower bound,   the literature on  these problems in the classical online model is rather scarce in comparison with scheduling, routing, and packing problems. In order to obtain meaningful results, one often needs to restrict inputs to  a special class of graphs or change the model of computation.

In this paper, we have considered geometric intersection graphs. They have various practical applications in wireless sensors and network routing~\cite{ClarksonV07,Har-PeledQ15}. For a family ${\cal S}$ of geometric  objects in $\IR^d$, the geometric intersection graph $G$ of ${\cal S}$ is an undirected graph with set of vertices same as ${\cal S}$ and the set of edges is defined as $E=\{(u,v)| u,v\in {\cal S}, u\cap v\neq \emptyset \}$. 
For example, if $\cal S$ is a family of disks then the corresponding graph is known to be a disk graph.  


\subsection{Model of computation}
 Similar to Boyar et al.~\cite{BoyarEFKL19}, we consider vertex arrival model. In the context of geometric intersection graphs,  the input objects are not known apriori, rather they are revealed one by one in discrete time steps. An online algorithm needs to take irrevocable decision after each time step without knowing the future part of the input, it needs to maintain a valid dominating set. As in~\cite{Eidenbenz}, in case of CDS, not only the newly revealed object, the online algorithm can add any previously seen objects to the existing solution set to make the graph connected. We analyze the quality of our online algorithm by competitive analysis~\cite{BorodinE}, where 
competitive ratio is the measure of comparison. In the following we have defined competitive ratio with respect to a minimization problem. For a maximization problem, one can define it in a similar way.  An online algorithm $ALG$ for a minimization problem is said to be \emph{{$\alpha$-competitive}}, if there exists a constant $\beta$ such that for any input sequence $\sigma$, we have 
$| ALG_\sigma|  \leq \alpha \times | OPT_\sigma|  +\beta$, where $ALG_\sigma$ and $OPT_\sigma$ are  the cost of the solution produced by the algorithm $ALG$ and  an optimal offline algorithm, respectively. The smallest $\alpha$ for which $ALG$ is $\alpha$-competitive is known as \emph{{asymptotic competitive ratio}} of $ALG$. The smallest $\alpha$ for which $ALG$ is $\alpha$-competitive with  $\beta$= $0$ is called the \emph{{absolute competitive ratio}} of $ALG$~\cite{ChrobakDFN20}. Throughout the paper, if not mentioned explicitly, by competitive ratio we mean absolute competitive ratio.

\subsection{Related Works}

The dominating set and its variants are well studied in the offline setup for graphs as well as geometric objects.
Finding MDS  is already known to be NP-hard~\cite{10.6/574848,LenzenW10} in the offline setup even for unit disks and unit squares. Polynomial time approximation scheme (PTAS) is known when the objects are homothets of a convex objects~\cite{DeL16}. 
King and Tzeng~\cite{KingT97} initiated the study in the context of online version for MDS. They showed that for general graph, the greedy algorithm achieves a competitive ratio of $n-1$, which is also tight bound  achievable by any online algorithm for the MDS problem, where $n$ is the number of vertices of graph. Even for interval graph, the lower bound of the competitive ratio is $n-1$ and   the simple greedy algorithm  achieves this bound~\cite{KingT97}. Eidenbenz~\cite{Eidenbenz} and Boyar et al.~\cite{BoyarEFKL19} studied DS, CDS and IDS for specific graph classes: trees, bipartite graphs, bounded degree graphs, planar graphs. Their results are summarized in Table 2 of~\cite{BoyarEFKL19}. For tree, Eidenbenz~\cite{Eidenbenz} and Boyar et al.~\cite{BoyarEFKL19} gave an upper bound on the competitive ratio as three for MDS. Later, Kobayashi~\cite{Kobayashi17} proved that three is also the lower bound for tree. Eidenbenz~\cite{Eidenbenz} proved that greedy algorithm achieves the tight bound of five for MDS of unit disk graph. For CDS, they showed that greedy algorithm achieves competitive ratio of $8+\epsilon$ whereas the lower bound is $10/3$. In this paper, we have studied the performance of the greedy algorithm for  other geometric intersection graphs. 

The maximum independent set problem is also an NP-hard problem even for  graphs with bounded degree 3~\cite{BermanF99}.   In the geometric setup, when the objects are pseudo disks PTAS is known  due to Chan and Har{-}Peled~{\cite{ChanH12}}.
The study of online versions of IS was initiated due to scheduling of intervals~\cite{LiptonT94}. 
In the classical online setup, even for interval graphs, it is known that no algorithm can achieve a competitive ratio  better that $n-1$, where $n$ is the number of nodes~\cite{OliverHKSV14, LiptonT94}. As a result, researchers have studied this problem in several other variants of online models.  Halldorsson et al.~\cite{HalldorssonIMT02}  studied the MIS problem for  graphs in  two different models. In the first model, the algorithm can maintain a multiple number
$r(n)$ of solutions and choose the largest one as the final one, where $n$ is the number of vertices. They proved that the best competitive ratio for this model is $\theta(\frac{n}{log n})$ when $r(n)$ is a polynomial, and $\theta(n)$ when $r(n)$ is a constant. Then they introduced a more powerful model where the algorithm can copy
intermediate solutions and extend the copied solutions in different ways.
In this model, they showed that the problem attains an upper
bound $O(\frac{n}{log n})$ and a lower bound $\Omega(\frac{n}{log^3 n})$ when $r(n)$ is a polynomial, and it achieves a tight bound of $\theta(n)$ when $r(n)$ is a constant. 
Due to the devastating lower bound of $\Omega(n)$ in the worst case analysis,  Göbel et al.~\cite{OliverHKSV14} 
 studied the online independent set via stochastic analysis.
 With stochastic analysis, they proved that graphs, with bounded inductive independence number $\rho$,  has competitive ratio $4c^3\rho^2$, where $c\geq 1$. Note that the inductive independence number for  geometric intersection graph of translates of a convex object in $\IR^2$ is $3$~\cite{YeB12}. Ye and Borodin~\cite{YeB12} conjectured that  the inductive independence number is $5$ for geometric intersection graph of translates of a convex object  in $\IR^3$. Therefore, the result in~\cite{OliverHKSV14} is at least $36$ and $100$ competitive for translates of a convex object in $\IR^2$ and $\IR^3$, respectively.
Göbel et al.~\cite{OliverHKSV14} also studied the weighted version of the MIS problem and proved using stochastic analysis that 
one can achieve $O(\rho^2\log n)$ competitive ratio, whereas   they have proved that the lower bound of the problem is $\Omega(\frac{\log^2 \log n}{\log n})$ even for weighted interval graphs. In this paper, we have studied the worst case analysis of the MIS problem for geometric intersection graphs in the classical online model.
\subsection{Notation and Preliminaries}
We use $[n]$ to denote the set $\{1,2,\ldots,n\}$. We use $\mathbb{Z}_n$ to denote the set $\{0,1,\ldots,n-1\}$ of non-negative integers less than $n$, where the  arithmetic operations such as addition, subtraction and multiplication are modulo $n$. By a \emph{geometric object}, we refer to a compact convex set  in $\IR^d$ with non-empty interior.  The \emph{center} of a geometric object  is defined as the centre of the smallest radius ball inscribing that geometric object. The \emph{neighbourhood} $N(u)$ of a geometric object $u$, belonging to a family $\cal S$ of geometric objects, is defined as the region that contains all the centers of objects  of $\cal S$ that  touch or intersect $u$. A geometric object whose center lies in $N(u)$ is referred as a \emph{{neighbour }} of $u$. Two objects are \emph{{congruent}} if one can be transformed into the other by  a combination of rigid motions, namely a translation, a rotation, and a reflection.
Two geometric objects are said to be \emph{{non-overlapping}} if they have no common interior, whereas we call them \emph{{non-touching}} if their intersection is empty. For a convex object $u \subset \IR^d$, \emph{{Newton number}} (a.k.a. \emph{{Kissing number}}), denoted as $K(u)$, is the maximum number of non-overlapping congruent copies of $u$ that can be arranged around $u$ so that each of them touches $u$~\cite{DBLP:journals/tcs/DumitrescuGT20,BrassMP}.
Some known bounds are as follows: $K(P_3)=12, K(P_4)=8$ and 
  $K(P_n)=5$, where $P_n$ is a regular polygon with  $n \geq 5$~\cite{article};
 $K(B_2)=6$, $K(B_3)=12$, where $B_2$ and $B_3$ are balls in $\IR^2$ and $\IR^3$, respectively~\cite{BrassMP,DBLP:journals/tcs/DumitrescuGT20, Leech,Schutte}.
 The kissing number for hyper-cubes is at least $3^d-1$, and whether this bound is sharp or not is still not settled. Throughout the paper, we have used $opt_{cds}$ to denote the size of the optimal connected dominating set in the offline setup. 

\begin{table}[!h]
\tiny
\centering
\setlength{\arrayrulewidth}{0.25mm}
\setlength{\tabcolsep}{6pt}
\renewcommand{\arraystretch}{2}
\begin{tabular}{|p{3.25cm}|p{1.25cm}|p{2.25cm}|p{1.75cm}|p{1.65cm}|p{1.45cm}|p{1.15cm}|}
\hline
  \multicolumn{1}{|c|}{} &  \multicolumn{1}{c|}{$K(U)$} & \multicolumn{1}{c|} {} &
 \multicolumn{1}{c|} {$\zeta$, DS, IDS, IS} & \multicolumn{3}{c|}{CDS}\\
\hline \hline
  Types of Geometric Objects &  Bounds & Previously known upper bound for IS $4c^3\rho^{2\$}$~\cite{OliverHKSV14} & Bounds   & Asymptotic competitive ratio of $GCDS$ & Lower Bound when $^*opt_{cds}=1$ & Lower Bound otherwise\\
\hline
\hline
Fixed oriented unit squares  &\cellcolor{gray}$8$~\cite{Youngs} & \cellcolor{gray}$36c^3$~\cite{OliverHKSV14,YeB12}   & $4$ & $6$ & $6$ & $\frac{7}{3}$ \\
\hline
Fixed oriented  unit hyper-cubes in $\IR^d$ & \cellcolor{gray}($3^{d}-1$,-)  & \cellcolor{gray}$4c^3(2d-1)^{2\#}$~\cite{OliverHKSV14,YeB12} & $2^d$ &  ($2^{d+1}-2$) & $7.2^{d-2}-1$ & $\frac{7.2^{d-2}}{3}$ \\
\hline
  Arbitrary oriented unit squares &\cellcolor{gray} $8$~\cite{Youngs} & $\_\_$ & $\{6,7\}$ & $12$ & $10$ & $\frac{11}{3}$ \\
  \hline
 Arbitrary oriented  unit hyper-cubes in $\IR^d$ & \cellcolor{gray}$(3^{d}-1$,-) & $\_\_$ & $[3.2^{d-1}$,-) & \_\_  & $11.2^{d-2}-1$ &  $\frac{11.2^{d-2}}{3}$ \\
 \hline
\rowcolor{gray}
\cellcolor{white}Unit disks & $6$~\cite{BrassMP,article} &\cellcolor{gray}$36c^3$~\cite{OliverHKSV14,YeB12}& $5$~\cite{Eidenbenz} & $6.906$~\cite{Eidenbenz,1512873} & $8$ & $3$ ~\cite{Eidenbenz}\\ \hline
Unit balls &\cellcolor{gray} $12$~\cite{BrassMP} &  \cellcolor{gray}$100c^{3\#}$~\cite{OliverHKSV14,YeB12}  & $12$ & $22$ & $17$ & $6$\\ \hline
Fixed oriented  unit triangles & \_\_ & \cellcolor{gray}$36c^3$~\cite{OliverHKSV14,YeB12}  & $\{5,6\}$ & $ 10$ & $9$ & $\frac{10}{3}$\\  \hline
Regular $k$-gons (fixed oriented) &\cellcolor{gray} $6$~\cite{article} & \cellcolor{gray}$36c^3$~\cite{OliverHKSV14,YeB12}  & $\{5,6\}$ & $ 10$ & $8$ & $3$\\ \hline
Arbitrary oriented Regular $k$-gons ($5\leq k \leq 10$) & \cellcolor{gray}$6$~\cite{article} & \_\_ & $6$ & $ 10$ & 10 & $\frac{11}{3}$ \\
\hline
\end{tabular}
  \begin{flushleft}
  \item \tiny{$^\$$ $c\geq 1$ and $\rho  $ is inductive independent number. }
  \item \tiny{*$opt_{cds}$ is the size of an optimal solution for CDS in the offline setup.}
  \item \tiny{$^{\#}$ Ye and Borodin~\cite{YeB12} conjectured that  the inductive independent number $\rho$ of translates of a convex object  in $\IR^d$ is $2d-1$. }
 \end{flushleft}
\caption{ Summary of Results: previously known results are marked in gray.}

\label{table_contri}
\end{table}

\normalsize

\subsection{Our Contributions}

First, 
similar to the kissing number, we introduce independent kissing number~$\zeta$ (see Definition~\ref{def1}) for a family of geometric objects, as well as for  graphs. Note that for a set of congruent objects, the value of independent kissing number is bounded from above by its kissing number. The independent kissing number is a fixed constant for several families of geometric objects: translated and/or rotated copies of a convex object in $\IR^d$, fat objects etc. Of course, the value depends upon $d$.

Next, for any  graph with independent kissing number~$\zeta$, we prove that the well known online greedy algorithm for dominating set achieves a competitive ratio of~$\zeta$, and any deterministic online algorithm can achieve a competitive ratio  at most~$\zeta$. This implies that 
the greedy algorithm for dominating set is optimal for any graph in the classical online setup. Note that the online algorithm need not know the value of~$\zeta$.
Since the greedy algorithm reports independent dominating set, the same result holds for independent dominating set for geometric intersection graphs in the online setup. 

For the maximum independent set problem, we prove that the same greedy algorithm achieves an optimal competitive ratio of~$\zeta$ for  graphs with independent kissing number~$\zeta$.  This result improves the best known upper bound of the problem for several geometric intersection graphs. The comparison is given in Table~\ref{table_contri}.

Next, we prove that,  for any  graph with independent kissing number~$\zeta$, the online greedy algorithm for connected dominating set achieves an asymptotic competitive ratio at most $2(\zeta-1)$. For a geometric intersection  graph of   translated copies of a convex objects in $\IR^2$  with independent kissing number $\zeta$, we prove that CDS has a lower bound of $\frac{2\zeta-1}{3}$.

As a consequence, the independent kissing number~$\zeta$  becomes an interesting  parameter of a graph. To get the idea of the value of~$\zeta$ for specific families of geometric objects,  we consider fixed oriented unit hyper-cubes in $\IR^d$, congruent balls in~$\IR^3$, fixed oriented unit triangles and regular polygons in~$\IR^2$. 
We show that the value of~$\zeta$ for  fixed oriented unit hyper-cubes in $\IR^d$ is $2^d$. 
It happens to be that for congruent balls  in~$\IR^3$, the value of~$\zeta$ is~12 which is also  same as its  kissing number. For both  fixed oriented unit triangles and fixed oriented regular polygons in~$\IR^2$, we  show that $5\leq \zeta\leq 6$. For each of these specific families, we also study the lower bound of CDS problem.  Our results are summarized in Table~\ref{table_contri}.

\subsection{Organization}

In Section~\ref{a:1}, we introduce the independent kissing number. Next,  in Section~\ref{sect:greedy}, we have discussed the performance of well known online greedy algorithms for MDS, MIS and MCDS problem for graphs with independent kissing number $\zeta$. Then in Section~\ref{sec:lb_MCDS}, we propose a lower bound of MCDS problem for translated copies of a convex object in $\IR^2$. After that, in Sections~\ref{3}-\ref{7}, we have discussed the value of $\zeta$ and lower bound of MCDS for specific families of geometric intersection graphs. Finally, in Section~\ref{10}, we give a conclusion.

\section{ Independent Kissing Number}\label{a:1}

First, we introduce \emph{{independent kissing  number}} $\zeta$ for a family  $\cal S$ of geometric objects in $\IR^d$.

\begin{definition}[Independent Kissing Number]\label{def1}
Let $\cal S$ be a family of geometric objects, and 
let $u$ be any object belonging to the family ${\cal S}$.  Let $\zeta_{u}$ be the maximum number of pairwise non-touching objects in $\cal S$ that can be arranged in and around $u$ such that all of them are dominated/intersected by $u$. The independent kissing number  $\zeta$ of  $\cal S$ is defined to be $\max_{u\in {\cal S}}\zeta_u$.
\end{definition}

 A set $K$ of objects belonging to the family $\cal S$ is said to form an \emph{independent kissing configuration} if (i) there exists an object $u\in K$  that intersects all the objects in $K\setminus\{u\}$, and (ii) all the objects in   $K\setminus\{u\}$ are  mutually non-touching to each other. Here $u$ and $K\setminus\{u\}$ are said to be  the \emph{core} and \emph{independent set}, respectively, of the independent kissing configuration. The  configuration is  considered \emph{optimal} if $| K\setminus\{u\}| =\zeta$, where $\zeta$ is the independent kissing number of $\cal S$.
The  configuration is said to be \emph{standard} if all the objects in $K\setminus\{u\}$ are mutually non-overlapping with $u$, i.e., their common interior is empty but intersection is non-empty. For illustration,  see Figure~\ref{fig:Independent kissing}.

\paragraph{Kissing Number vs Independent Kissing Number}
Note that the kissing number is defined mainly for  a set of congruent objects. On the other hand, the  independent kissing number is defined for any set of geometric objects. In kissing configuration, the objects around $u$ are non-overlapping but they are allowed to touch each other; on the other hand, in the independent kissing configuration, the objects around $u$ are non-touching. For more illustration, see Figure~\ref{fig:configuration}. 
For a set of congruent copies of $u$, it is easy to observe that  the value of independent kissing number is at most $K(u)$, the kissing number of $u$.

\begin{figure}[h]
  \centering
     \begin{subfigure}[b]{0.45\textwidth}
         \centering
\includegraphics[scale=0.55]{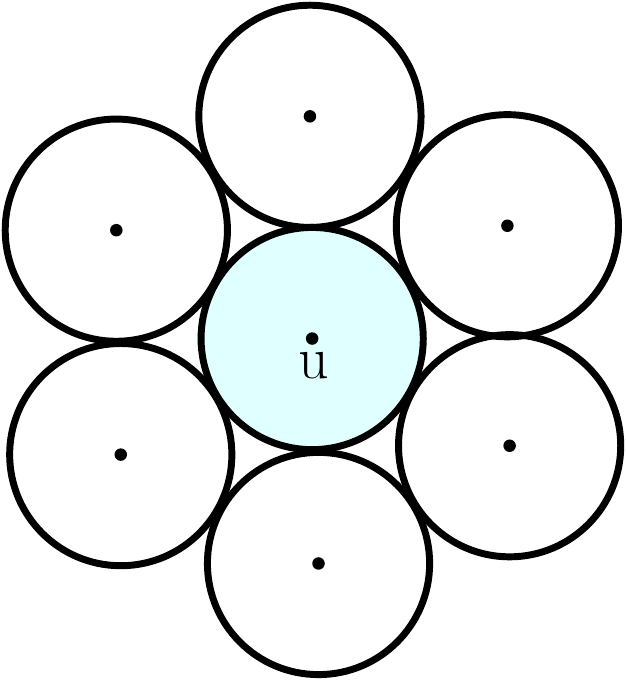} 
\caption{}
\label{fig:pentagon}
     \end{subfigure}
     \hfill
     \begin{subfigure}[b]{0.45\textwidth}
         \centering
\includegraphics[scale=0.55]{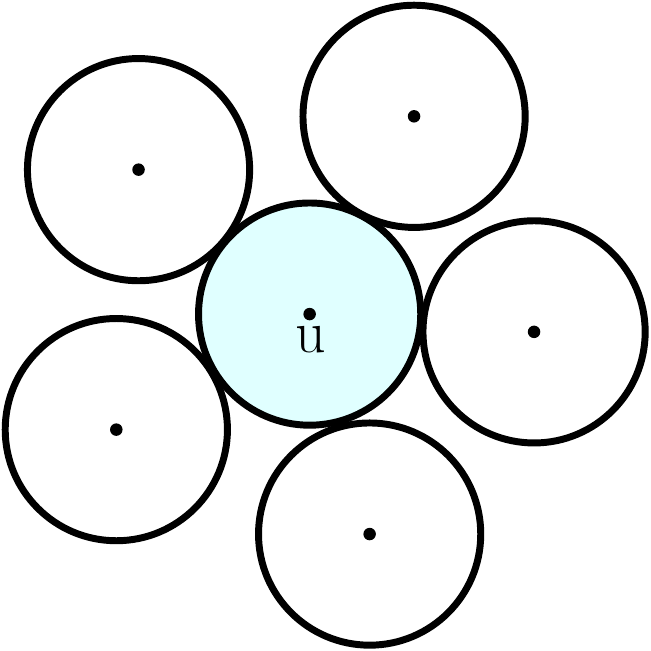}
\caption{}
\label{fig:Independent kissing}
     \end{subfigure}
       \caption{(a) Optimal kissing configuration, and   (b) optimal independent kissing configuration for unit disks.}
       \label{fig:configuration}
\end{figure}

\paragraph{Graph Theoretic View}  
Now, we give an alternative definition of independent kissing number in terms of graph. Let $G=(V,E)$ be a graph.
 For any vertex $v\in V$,  we denote $N(v)=\{u(\neq v) \in V | (u,v) \in E\}$ as the neighbourhood of the vertex $v$. For any subset $V'\subseteq V$, the \emph{induced subgraph} $G[V']$ is the graph whose vertex set is $V'$ and the edge set consists of all of the edges in $E$ that have both endpoints in  $V'$.
 By $\alpha(G)$, we denote  the size of a maximum independent set of a graph $G$. The independent kissing number $\zeta$ of a graph $G$ is defined as $max_{v \in V} \{\alpha(G[N(v)])\}$.


\section{Online Greedy Algorithms For Graphs }
\label{sect:greedy}
In  this section, we discuss popularly known greedy  online algorithms for MDS, MIS and MCDS  for graphs.  We show how their performance depend on the independent kissing number $\zeta$. Note that the algorithms need not know the value of $\zeta$ in advance.
\subsection{Greedy Algorithm for Dominating Set ($GDS$)}\label{2.1}
 Suppose that, after observing first $t-1$ vertices, we have a solution set $\A_{t-1}$.  Let $v_t$ be a new   vertex.
\begin{description}
  \item[$\bullet$ ] If $v_t$ is dominated by $\A_{t-1}$, then we set  $\A_t=\A_{t-1}$.
  \item[$\bullet$ ] Otherwise,  we will add $v_t$ into our solution set by doing $\A_t=\A_{t-1}\cup \{v_t\}$.
\end{description}

\begin{observation}\label{obs1}
The  set of vertices returned by the  {$GDS$ algorithm} are
pairwise non-adjacent. In other words, the solution set is always  an   independent set. 
\end{observation}

As a result of Observation~\ref{obs1}, the $GDS$ algorithm achieves the same competitive ratio for both MDS and MIDS problems.
\paragraph{Upper Bound of the Algorithm}

Let $G$ be any graph with  independent kissing number $\zeta$.
 Let $\A_t$ be a set of vertices   reported by our $GDS$ {algorithm} after receiving first $t$ vertices $v^t=\{v_1,\ldots,v_t\}$. 
Similarly, let $\OO_t$ be the set of vertices in the optimum set for $v^t$. 
 Since $\OO_t$ also dominates $\A_t$, any vertex $A\in \A_t$  will either belong to the optimal set $\OO_t$ or will be the neighbour of a vertex in $\OO_t$. If $\A_t \cap \OO_t \neq \emptyset$, then remove all the common elements from $\A_t$ and  $\OO_t$ resulting, $\A'_t \cap \OO'_t = \emptyset$ ($\frac{| \A'_t| }{| \OO'_t| } \leq M$ implies $\frac{| \A_t| }{| \OO_t| } \leq M$). Without loss of generality, assume that $\A_t \cap \OO_t = \emptyset$.
 Thus, each vertex $A\in \A_t$ will be in the neighbour of at least a vertex in $\OO_t$. Note that all the vertices in $\A_t$ are pairwise non-adjacent (Due to Observation~\ref{obs1}). From the definition of independent kissing number (Definition~\ref{def1}), we know that each vertex $O\in\OO_t$ can have at most $\zeta$  non-adjacent  vertices of $\A_t$ in its neighbourhood. Therefore, $| \A_t| \leq \zeta| \OO_t| $, and our algorithm achieves a competitive ratio of at most~$\zeta$.

\noindent
\paragraph{Lower Bound of the Problem}
 Let $v_1,v_2,\ldots,v_{\zeta}$ \& $v_{\zeta+1}$ be the vertices coming one by one to the online algorithm. Since $\zeta$ is the independent kissing number of the  graph under consideration,
with an adaptive deterministic adversary, we can choose vertices   $v_1,v_2,\ldots,v_{\zeta}$ such that all of them are pair-wise non-adjacent and can be dominated by a single vertex $v_{\zeta+1}$. Thus, any online algorithm will report first $\zeta$ vertices as dominating set, whereas the  optimum dominating set contains only the last vertex $v_{\zeta+1}$.  Thus the lower bound on the competitive ratio for minimum dominating set is~$\zeta$.
Therefore, we have the following.
 \begin{theorem}\label{th:mds_geometric_objects}
For a graph $G$ with independent kissing number $\zeta$, the  greedy dominating set ($GDS$) algorithm is optimal  and it achieves a competitive ratio of $\zeta$.
\end{theorem}

\subsection{Maximum Independent Set}\label{9}

Recall that the algorithm $GDS$ defined in Section~\ref{2.1}, produces a dominating set which is also an independent set for a graph $G$. Here, we prove that $GDS$ 
 yields a competitive ratio  $\zeta$ matching the lower bound  for the maximum independent set problem as well. 

\begin{theorem}
There exists a deterministic online algorithm having competitive ratio at most $\zeta$ for maximum independent set of a graph $G$ with  independent kissing number $\zeta$. This result is tight: the competitive ratio of any
deterministic online algorithm for this problem is at least $\zeta$.
\end{theorem}
\begin{figure}[h]
\centering
\includegraphics[width=100mm]{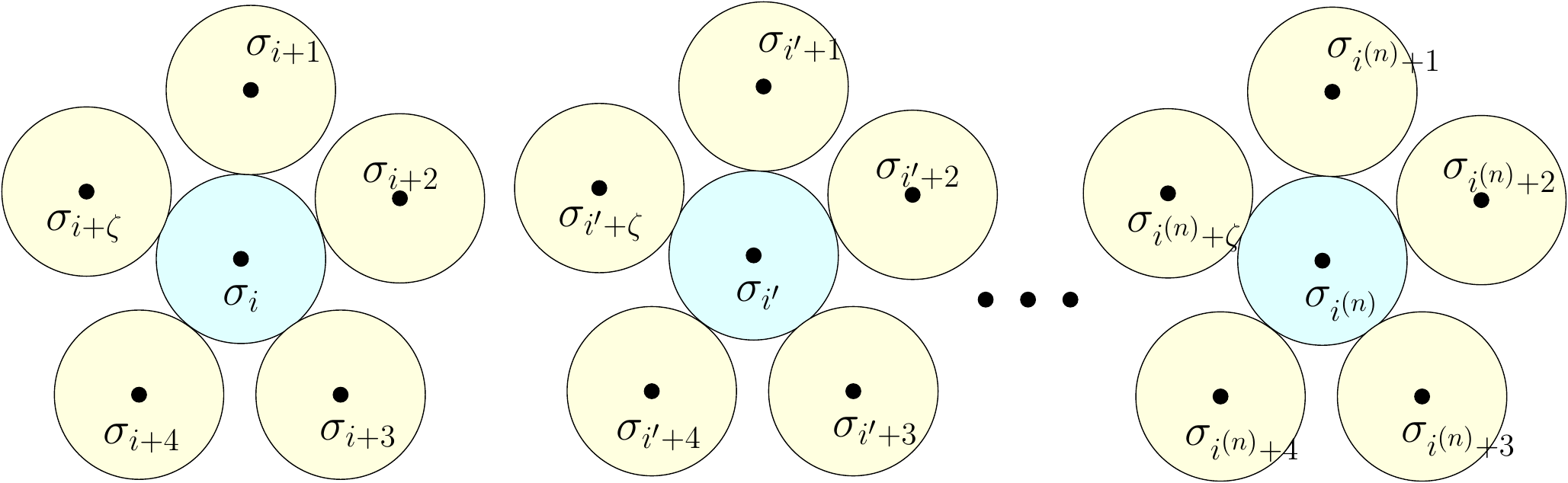}
\caption{Input instance of unit disks for maximum independent set. From each cluster Alice will add Yellow colored objects as the input instance only if Bob chooses Blue colored object.}
\label{fig:IS}
\end{figure}
\begin{proof}
Let  $\A_i$ be a set of vertices   reported by our $GDS$ {algorithm} after receiving first $i$ vertices $v^i=\{v_1,\ldots,v_i\}$. 
Similarly, let $\OO_i$ be the set of vertices in the optimum set for $v^i$. It is easy to observe that $| \A_i|\leq | \OO_i| $.  Let $N_{\OO_i}(a_i)$ be the set of vertices in $\OO_i$ that are adjacent to $a_i \in A_i$. Now, consider any vertex $o_i $ and $a_i $ from $\OO_i$ and $A_i$, respectively. By the definition of $\zeta$ (independent kissing number), we can observe that $| N_{\OO_i}(a_i)|  \leq \zeta$ (For each $a_i$ there are at most $\zeta$ vertices from $\OO_i$ that are adjacent to  $a_i$). We know, $A_i$ is a dominating set for $\OO_i$, so it is easy to observe that $\cup_{a_i \in A_i} N_{\OO_i}(a_i) = \OO_i$.
From above, we can conclude that the upper bound of the proposed $GDS$ {algorithm} for maximum independent set of a graph $G$ has a competitive ratio~$\zeta$.


To proof the lower bound for maximum independent set, we can think of  a game played between two players: Alice and Bob. Here, Alice plays the role of the adversary and Bob plays the role  of the online algorithm. Alice  presents the vertices of a graph $G$ one by one to Bob, and  Bob needs to decide whether to consider that vertex in the independent set or not. Alice have to present vertices in such a way that will force Bob to choose as less new vertices as possible; but Bob needs to choose as many vertices as possible to win the game.  Let  $\zeta$ be the independent kissing number of the graph under consideration and $v_i$ is the $i^{th}$ vertex. Now, Alice presents $v_i$ to Bob;  if Bob ignores it then Alice will present some $v_{i'}$ to Bob where, $v_{i'}$ is far away from all previous vertices.  
So without loss of generality, assume Bob chooses $v_i$ then Alice will present $v_{i+1},\ldots,v_{i+\zeta}$ as the sequence of vertices coming one by one to Bob such that all of them are pair-wise non-adjacent vertices {but each of them are adjacent to $v_i$} (See Figure~\ref{fig:IS}). Since the independent kissing number of the graph under consideration is $\zeta$, it is always possible to construct such a sequence of vertices. Now, any optimal offline algorithm will  choose all these $\zeta$ vertices  $v_{i+1},\ldots,v_{i+\zeta}$ in the independent set and Bob will not be able to choose any of them. Therefore, the lower bound on the competitive ratio for maximum independent set of graph $G$ in the online setup is~$\zeta$.
\end{proof}

\subsection{Greedy Algorithm for Connected Dominating Set ($GCDS$)}\label{2.2}
 Eidenbenz~\cite{Eidenbenz} proposed that the following  greedy algorithm for online connected dominating set problem achieves $8+\epsilon$ competitive ratio for  unit disk graph.
 We generalize this  for    graphs with fixed independent kissing number $\zeta$.
 
 Suppose that, after  observing first $t-1$ vertices $v_1,\ldots, v_{t-1}$,  we have a solution set $\A_{t-1}$. Let $v_t$ be the new vertex.
\begin{itemize}
  \item[$\bullet$ ] If $v_t$ is dominated by  the set $\A_{t-1}$, then we  set  $\A_{t-1}$ as the $\A_t$.
  \item[$\bullet$ ] Otherwise,  we will add $v_t$ and one of its neighbour $v_j$ where $j\in [t-1]$ into our solution set. In other words, we set $\A_t=\A_{t-1} \cup \{v_t,v_j\}$.
  \end{itemize}

\begin{observation}\label{obs2}
Whenever the size of the solution set is increased (excepting for the first time when $v_1$ is inserted) the $GCDS$ algorithm adds two vertices $v_t$ and $v_j$ to the solution set, where $v_t$ is independent with $\A_{t-1}$. So,  instead of $\A_t$, if  we maintain two disjoint sets $I_{t}$ and $J_{t}$ such that we  insert $v_t$ to $I_{t}$ and $v_j$ to $J_{t}$, then  $I_{t}$ is an independent set. Since $| I_{t}| =| J_{t}| +1$, we have  $|\A_t| =2| I_{t}| -1$, where $I_{t}$ is an independent set.
\end{observation}

 Let $G(V,E)$ be a  graph with independent kissing number $\zeta$.
Now, we  generalize a result by Wan et al.~\cite[Lemma 9]{WanAF04} that is one of the main ingredients for analysing the upper bound.

\begin{lemma}\label{lm:cds_opt}
Let $\cal I$ and $\OO$ be any independent set and minimum connected dominating set, respectively,  of a  graph  with independent kissing number $\zeta$. Then $| {\cal I}| \leq (\zeta-1){| \OO| } +1$.
\end{lemma}
\begin{proof}
 Let $T$ be any spanning tree of $\OO$.
 Consider an arbitrary preorder traversal of $T$ given by $v_1, v_2,\ldots, v_{{| \OO| }}$. Let ${\cal I}_1$ be the set of nodes in ${\cal I}$ that are adjacent to $v_1$. For any, $2\leq i\leq {| \OO| }$, let ${\cal I}_i$ be the set of nodes in ${\cal I}$ that are adjacent to $v_i$ but none of $v_1,\ v_2,\ldots,\ v_{i-1}$. Then ${\cal I}_1,\ {\cal I}_2,\ldots,\ {\cal I}_{{| \OO| }}$ forms a partition of ${\cal I}$. As $v_1$ can be adjacent to at most $\zeta$ independent nodes (due to the definition of $\zeta$), $| {\cal I}_1|  \leq \zeta$.  For any, $2\leq i\leq {| \OO| }$, at least one node in $v_1,\ v_2,\ldots,\ v_{i-1}$ is adjacent to $v_i$ but ${\cal I}_i$ is adjacent to none of $v_1,\ v_2,\ldots,\ v_{i-1}$. 
So, $| {\cal I}_i|  \leq (\zeta-1)$ where $2\leq i\leq {| \OO| }$. Thus, we have 
\begin{equation*}
  | {\cal I}| = \sum_{i=1}^{{| \OO| }} | {\cal I}_i| \\
   = {\cal I}_1 +\sum_{i=2}^{{| \OO| }} | {\cal I}_i| \\
   \leq \zeta + ({| \OO| }-1)(\zeta-1)\\
   \leq (\zeta-1){| \OO| } +1.
  \end{equation*}
Hence  the lemma follows.
\end{proof}

\begin{theorem}\label{cds}
For a graph with independent kissing number $\zeta$, the $GCDS$ algorithm achieves the asymptotic competitive ratio $2(\zeta-1)$
and absolute competitive ratio $min\{2\zeta, 2(\zeta-1)(1+\frac{1}{| \A_t| -1})\}$ (i.e., $2(\zeta-1)+\epsilon$), where $\A_t$ is the solution returned by the $GCDS$ algorithm.
\end{theorem}
\begin{proof}

From Observation~\ref{obs2}, we have $| \A_t| =2| I_t| -1 \leq 2| I_t| $. On the other hand, from Lemma~\ref{lm:cds_opt}, we have $| I_t| \leq (\zeta-1)| \OO_t| +1$. As a result, we get $| \A_t| \leq 2(\zeta-1)| \OO_t| +2$. Therefore, the asymptotic competitive ratio of the $GCDS$ algorithm is $2(\zeta-1)$.

By Observation~\ref{obs2}, we have $| \A_t| =2| I_t| -1\leq 2| I_t| $, where $I_t$ is an independent set. On the other hand, from Lemma~\ref{lm:cds_opt}, we have $| { I_t}| \leq (\zeta-1){| \OO_t| } +1\leq \zeta {| \OO_t| }$ (the second inequality follows from the fact that ${| \OO_t| }\geq 1$).  
Thus, the competitive ratio is:
$\frac{| \A_t| }{{| \OO_t| }}\leq 
     \frac{2| {I_t}| }{\frac{| { I_t}| }{\zeta}} = 2\zeta$.
     Alternative (tighter) analysis is:
\begin{equation*}
      \frac{| \A_t| }{{| \OO_t| }}\leq 
     \frac{2| {I_t}| -1}{\frac{| { I_t}| -1}{\zeta-1}} = 2(\zeta-1)(1+\frac{1}{2| { I_t}| -2})= 2(\zeta-1)(1+\frac{1}{| { \A_t}| -1}). 
\end{equation*}
This completes the proof.
 \end{proof}


\paragraph{Improved Bound for Unit Disk Graph}
For unit disk graphs Funke et~al.~\cite{1512873} presented $| I_t|  \leq 3.453| \OO_t|  + 8.291$. Since $| \A_t|  \leq 2| I_t| $ (from Observation~\ref{obs2}) therefore, $| \A_t|  \leq 6.906| \OO_t| +16.582$. Hence,  we  conclude that the deterministic $GCDS$ algorithm  obtains asymptotic competitive ratio of at most $6.906$ for unit disk graphs. 


 \section{Lower Bound of MCDS Problem for Translated Copies}\label{sec:lb_MCDS}
Eidenbenz~\cite{Eidenbenz} proposed lower bound  of $\frac{10}{3}$ for connected dominating set of unit disks in the online setup by constructing an example. In this section, we have generalized   that idea for any translated convex objects in $\IR^2$.

First,  we prove that for a family of  translated convex objects, we can always have optimal independent kissing configuration in standard form. To prove this,  we need the following definition and observation.

\begin{definition}[Convex Distance $d_C$]
  Let $p_1,p_2$ be any two points in $\IR^2$ and let $C$ be a convex object. Then the convex distance function $d_C(p_1,p_2)$, induced by $C$, is  defined as the smallest $\lambda \geq 0$ such that $p_1, p_2 \in \lambda C$, while the center of $C$ is at $p_1$. Let $O$ be any object, then the point-to-object distance $d_C(p_1,O)$ is defined as $\min_{x\in O} d_C(p_1,x)$.
\end{definition}

\begin{observation}\label{obs:0}
Let $O$ and $L$ be a convex object and a line, respectively. Let $x$ be any point on $L$. If $x$ moves along the  line $L$ from one end to another, then   the distance $d_C(x,O)$ 
 first decreases monotonically until a minimum
   point (or an interval of the same minimum value), and then it increases
   monotonically. (Refer to Figure~\ref{fig:convex_dist}.)
\end{observation}

\begin{lemma}\label{lm:standardIndepndentKC}
 There  exists  standard optimal independent kissing configuration for a family of translated copies of a convex object in $\IR^2$.
\end{lemma}

\begin{figure}[!h]
  \centering
     \begin{subfigure}[b]{0.325\textwidth}
         \centering
\includegraphics[scale=0.35]{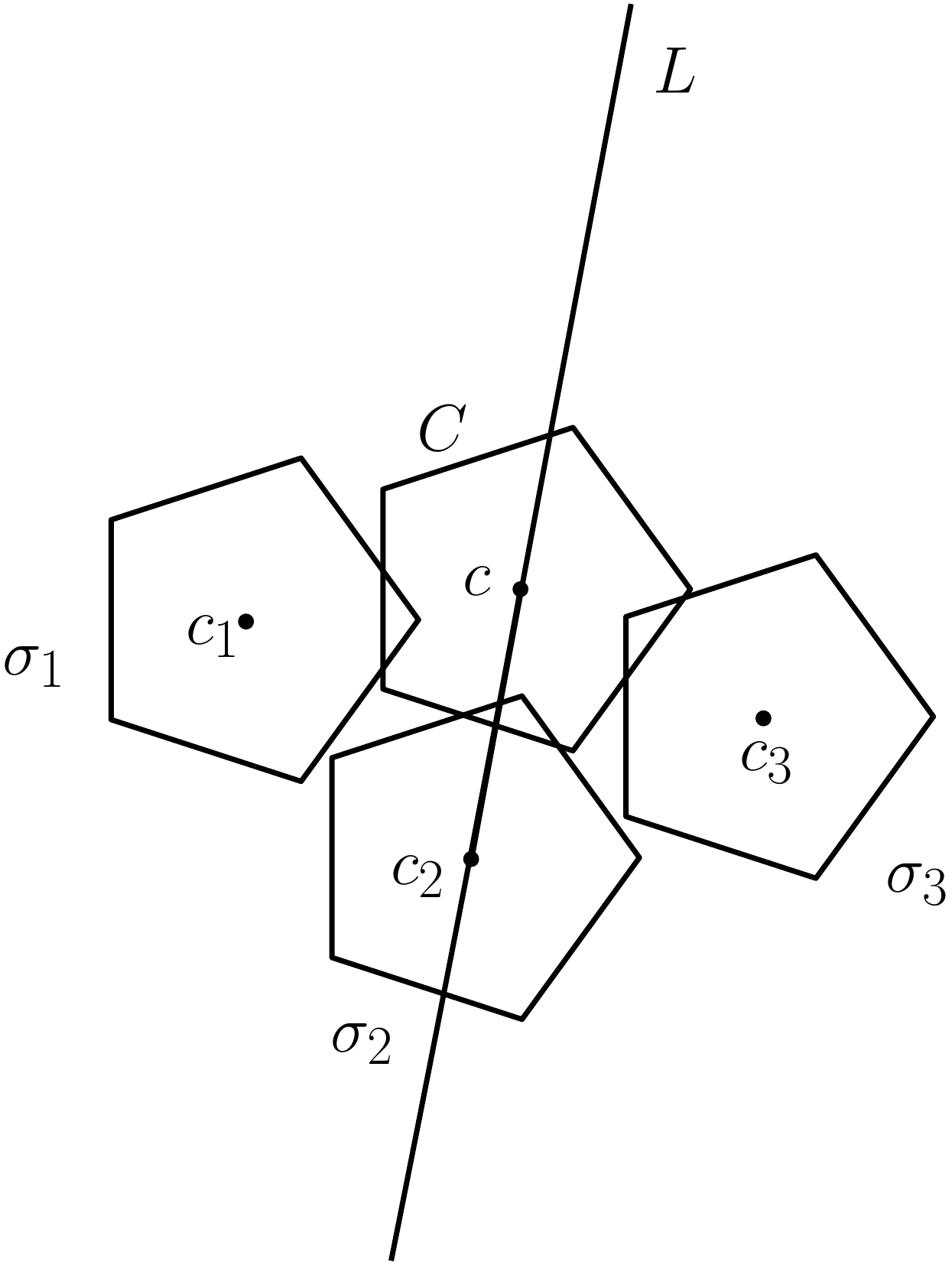} 
\caption{}
\label{fig:ikn_to_sikn}
     \end{subfigure}
     \hfill
      \begin{subfigure}[b]{0.325\textwidth}
         \centering
\includegraphics[scale=0.35]{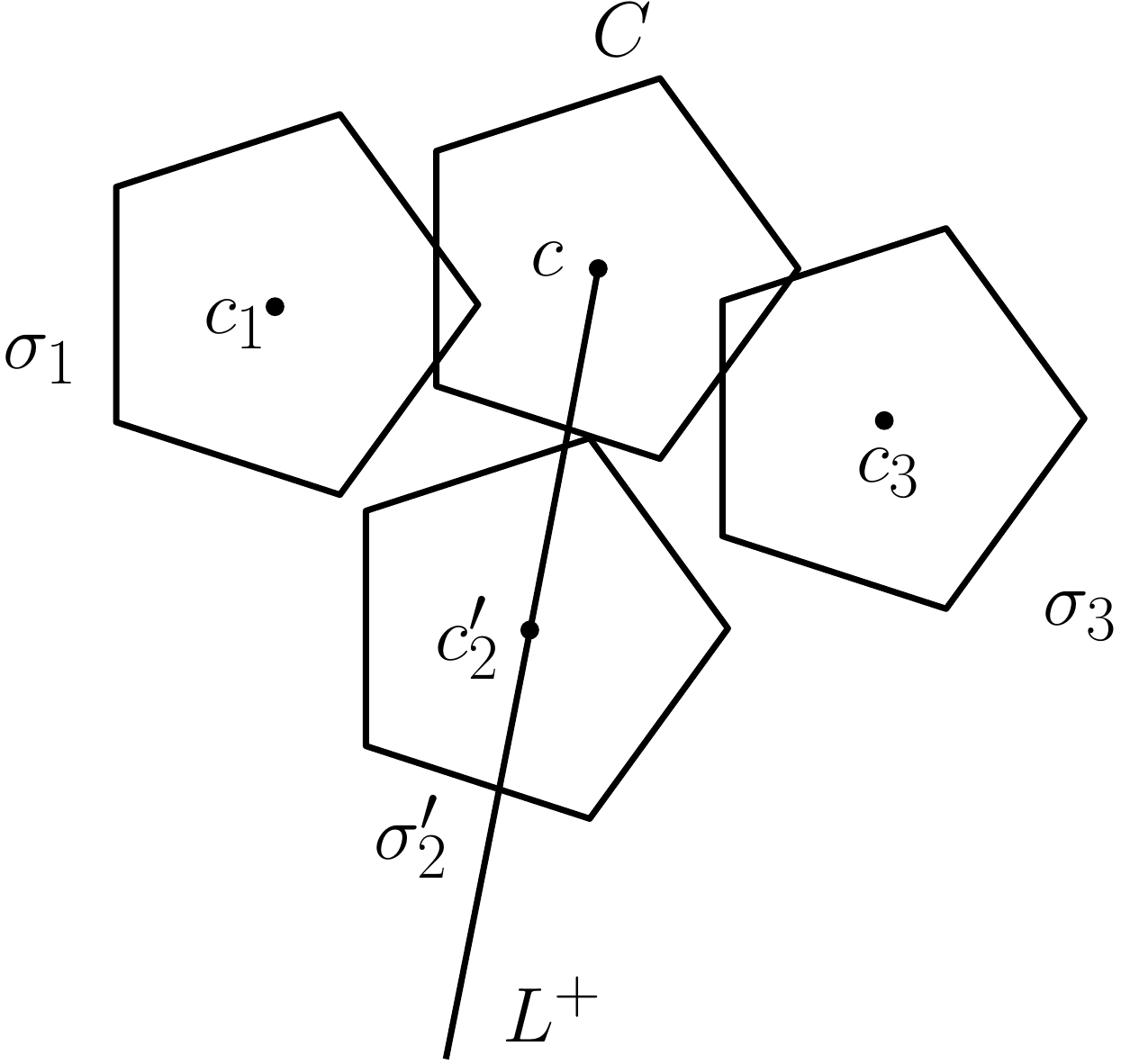}
\caption{}
\label{fig:move}
     \end{subfigure}
     \hfill
     \begin{subfigure}[b]{0.325\textwidth}
         \centering
\includegraphics[scale=0.35]{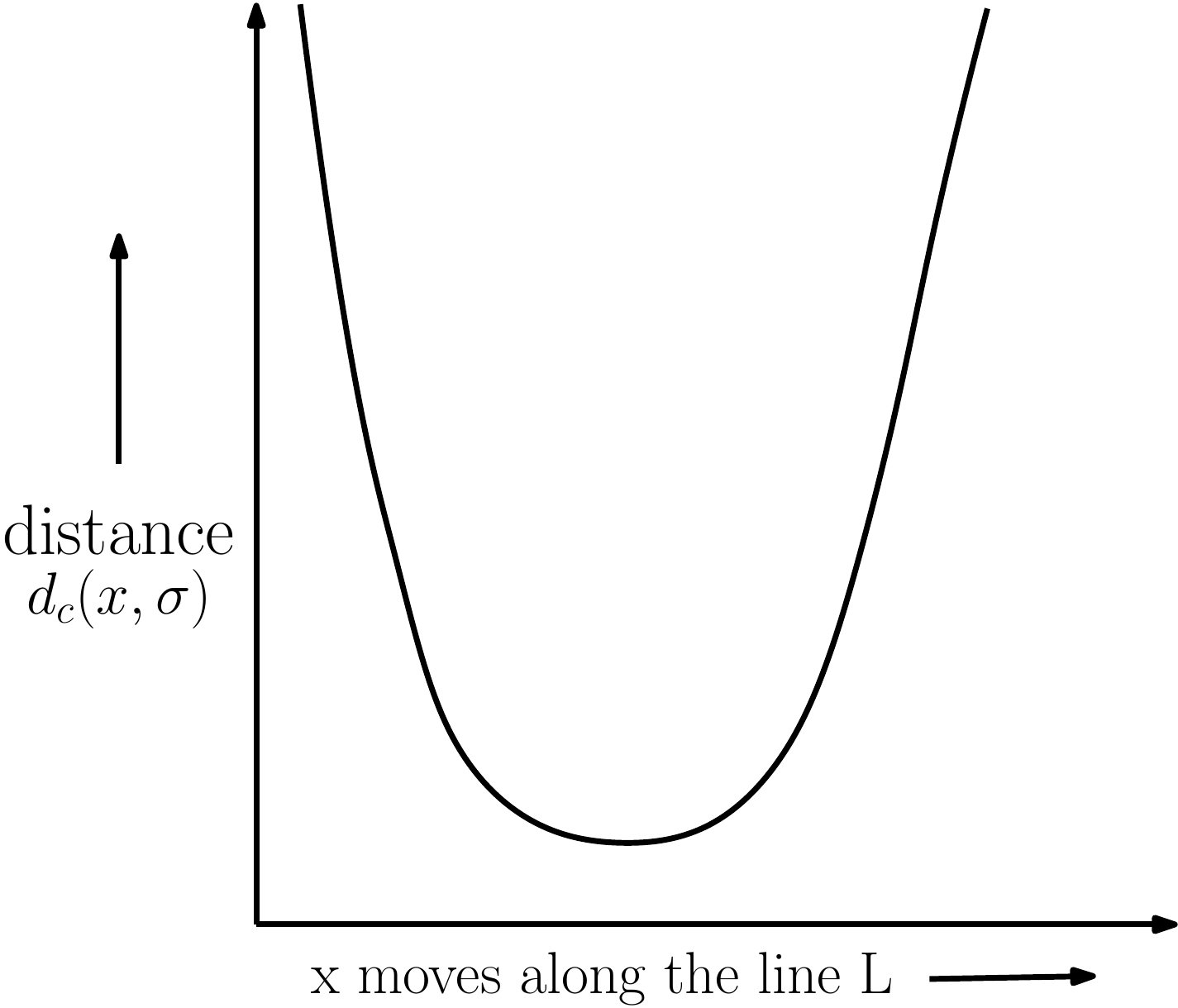}
\caption{}
\label{fig:convex_dist}
     \end{subfigure}
       \caption{(a-b)~Illustration of Lemma~\ref{lm:standardIndepndentKC}:  Moving $\sigma_2$ along $L^{+}$, (c)~Illustration of Observation~\ref{obs:0}.}
       \label{fig:translates}
\end{figure}

\begin{proof}

To prove this, we will show that given any optimal independent kissing configuration, we can transform it into a standard form. Consider an optimal independent kissing configuration where $C$ is the core  and $\cal I$ is the independent set. Let  $({\sigma}_0$, ${\sigma}_1$,\ldots, ${\sigma}_{\zeta-1})$ be  the anti-clockwise  order  of appearance of objects in $\cal I$  around the core object $C$ in the configuration, where ${\sigma}_i\in {\cal I}$, $i\in Z_{\zeta}$. Consider an object ${\sigma}_i\in {\cal I}$  that has non-empty common interior with $C$. Let $c$ and $c_i$  be the center of the object $C$ and ${\sigma}_i$, respectively  (see Figure~\ref{fig:translates}(a-b).). Let $L$ be the line obtained by extending the line segment $cc_i$ in both direction.  Since the objects ${\sigma}_i$ and ${\sigma}_{i-1}$ (resp., ${\sigma}_{i+1}$) are non-touching and ${\sigma}_i$ intersects $C$, we have  $d_C(c_i,{\sigma}_{i-1}) > d_C(c,{\sigma}_{i+1})$ (resp., $d_C(c_i,{\sigma}_{i+1}) > d_C(c,{\sigma}_{i+1})$). Note that the operation on index $i$ is modulo $\zeta$. Consider the half-line $L^{+}$  of $L$  that has one end point at $c_i$ and that does not contain the point $c$. Due to Observation~\ref{obs:0}, if $x$ moves along $L^{+}$  from $c_i$ then the distance $d_C(x, {\sigma}_{i-1})$ (resp., $d_C(x, {\sigma}_{i+1})$) monotonically increases. As a result, if we translate  the object ${\sigma}_i$ by moving the center along $L^{+}$, the translated copy ${\sigma}_i'$ never touches the object ${\sigma}_{i-1}$ (resp., ${\sigma}_{i+1}$). Hence, we will be able to find a point $x$  on $L^{+}$ such that the translated copy ${\sigma}_i'$ centering on $x$ will touch the boundary of $C$, but will not touch ${\sigma}_{i-1}$ and ${\sigma}_{i+1}$.  We replace the object ${\sigma}_i$ with ${\sigma}_i'$ in the independent kissing configuration. For each object $I\in {\cal I}$ that has non-empty common interior with $C$,
we follow the similar approach as above. This ensures that all these $\zeta$ objects are mutually non-overlapping (but touching) with $C$. Hence, the lemma follows.
\end{proof}

Now, we introduce some structural definitions. 

\begin{definition}[Path of cycles]
A graph $G(V,E)$  is said to be  a \emph{path of cycles} of length $m$ and size $k$ if the following three properties are satisfied: 
\begin{itemize}
    \item[1](cardinality) $G$ has exactly $mk$ vertices and $mk+m-1$ edges;
    \item[2](cycle-partition) $V$ can be partitioned into $m$ disjoint sets $V_1, V_2,\ldots V_m$ each having $k$ vertices such that each of the induced subgraph $G[V_i]$ is a cycle, where $i\in [m]$;
    \item[3](head and tail) for each $i \in [m-1]$, there exists exactly one edge $(h,t)\in E$  such that $h\in V_i$ and $t \in V_{i+1}$.  We refer $h$ and $t$ as the \emph{head} and \emph{tail} of $V_i$ and $V_{i+1}$, respectively. For $V_1$ (resp., $V_m$) any arbitrary vertex other than head (resp., tail) can be denoted as tail (resp., head).
\end{itemize}
\end{definition}
Refer to Figure~\ref{fig:new_cyclone} for an example of a path of cycles.
Let $V_1, V_2,\ldots, V_m$ be the cycle-partition of a path of cycles $G(V, E)$ (of length $m$ and size $k$) such that head of $V_i$ is adjacent to tail of $V_{i+1}$, where $1\leq i< m$. 
 Let $v_{0}, v_{1}, \ldots v_{(k-1)}$ be the clock-wise order of vertices  in a cycle $G[{V}_i]$. Let $v_h$ and $v_t$ be the head and tail of the cycle $G[V_i]$. Now, we define the cyclone-order of vertices in a cycle and a path of cycles.

\begin{figure}[h]
\centering
\includegraphics[width=140mm]{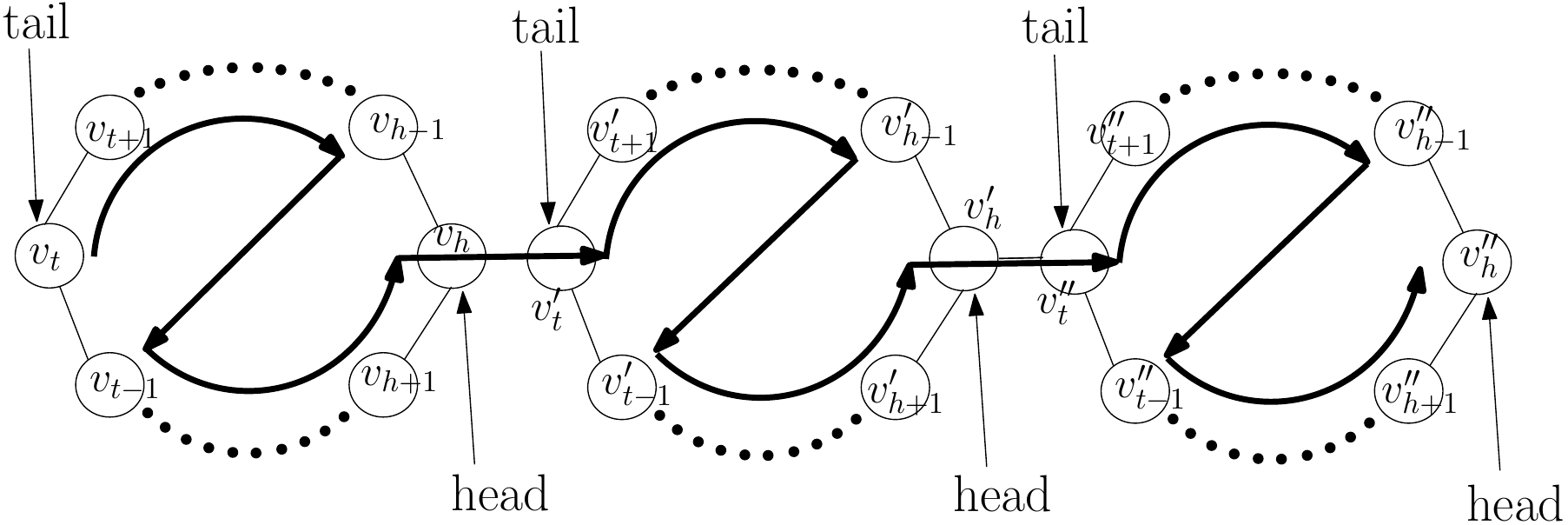}
\caption{Path of cycles and cyclone order: Head and tail of each cycle is marked. Cyclone order is indicated by bold arrow.}
\label{fig:new_cyclone}
\end{figure}

\begin{definition}[Cyclone-order of vertices in a cycle]
 We define the cyclone-order of the vertices of a cycle  $G[V_i]$ as  the clockwise enumeration  of $h-t$ vertices $[v_t, v_h)$: $v_t, v_{t+1}, \ldots, v_{h-1}$ followed by anti-clockwise enumeration of remaining $k-h+t$ vertices $(v_t, v_h]$: $v_{t-1}, v_{t-2},\ldots, v_{h}$.   We denote the first $h-t$ length clock-wise sequence as \emph{cw-part} and the remaining as \emph{acw-part} of a cyclone-order.

 \end{definition}
 
 \begin{definition}[Cyclone-order of vertices in a path of cycles]
 The cyclone-order of the vertices of a path of cycle is defined as $cyclone(V_1)$,  $cyclone(V_2), \ldots, cyclone(V_m)$,  where   $cyclone(V_i)$ represents the cyclone-order of the vertices in $V_i$, $i\in[m]$. 
 \end{definition}
See Figure~\ref{fig:new_cyclone} for an illustration of cyclone order of vertices in a path of cycles. 
Now, we prove an important property of cyclone-order traversal of vertices in a path of cycles. This property plays an important role in the lower bound construction.
Proof of the following Lemma is similar to Lemma~10 in~\cite{Eidenbenz}.
\begin{lemma}\label{lemma:path-of-cycles}
 For a path $\cal P$ of cycles of size $k$ and length $m$, if the vertices are enumerated in a cyclone order,  then any deterministic online algorithm reports  CDS of  size at least  $(k-1)m-1$.  
\end{lemma}
\begin{proof}
Let $V_1, V_2,\ldots, V_m$ be the cycle-partition of a path of cycles $G(V, E)$ (of length $m$ and size $k$) such that head of $V_i$ is adjacent to tail of $V_{i+1}$, where $1\leq i< m$. 
Consider an input sequence of vertices of $V_i$ where the vertices arrive in cyclone-order. First, let us consider the cw-part of the sequence  $[v_t, v_h)$: $v_t, v_{t+1}, \ldots, v_{h-1}$. 
 Any online algorithm will report at least first $h-t-1$ vertices as a CDS for this sequence. Next, if we append the first $k-(h-t)-1$ vertices of acw-part  $(v_t, v_h)$: $v_{t-1}, v_{t-2},\ldots, v_{h+1}$, then  any online algorithm will report at least  $k-3$  objects as CDS in total. Now, if we append $v_h$ to the sequence, then any online algorithm will either add $v_{h-1}$ or $v_{h+1}$ in the CDS.  Therefore, for each individual cycle, any online algorithm reports CDS of size at least  $k-2$.  Now, if we append the vertices of $V_{i+1}$ in cyclone order, then to make a connected dominating set, the algorithm needs to add $v_h$ of $V_i$. As a result, for each $1\leq i <m$, any online algorithm will  report $k-1$ vertices; for $V_m$ it will report $k-2$.  Hence the lemma follows.
\end{proof}

Now, we introduce the concept of C-block which is a basic structure for the construction of lower bound.

\begin{definition}[C-block]\label{def:C-block}
Let $\cal S$ be a family of geometric objects. 
A set $\cal B$ of   objects from $\cal S$ is said to form a \emph{C-block} if following two conditions are satisfied:
\begin{itemize}
    \item the geometric intersection graph of the objects in $\cal B$ is a cycle and
    \item  there exists an object $u(\notin {\cal B})$ in the family $\cal S$ that intersects all the objects in $\cal B$.
\end{itemize}  We refer the object $u$ as the \emph{core} of the  C-block $\cal B$. The \emph{size} of a C-block is defined as the  number of objects in $\cal B$. 
\end{definition}

\begin{lemma} \label{lm:C_block}
There exists a C-block of size $2\zeta$ for a family of translated copies of a convex object.
\end{lemma}
\begin{proof}
The proof is by construction. 
  Due to Lemma~\ref{lm:standardIndepndentKC}, we know that there exists a standard optimal independent kissing configuration for a family of translated copies of a convex object. Let $K$ be a standard independent kissing configuration where $C$ is the core object and $\cal I$ is the independent set. Let  $({\sigma}_0$, ${\sigma}_1$,\ldots, ${\sigma}_{\zeta-1})$ be  the clockwise  order  of appearance of objects in $\cal I$  around the core object $C$ in the configuration, where ${\sigma}_i\in {\cal I}$, $i\in \mathbb{Z}_{\zeta}$. Let us define a locus $\cal L$ that contains all points $x\in \IR^2$ such that $d_C(x, C)=1$. Note that the centers of all the objects $\sigma_i\in {\cal I}$ lies on the locus $\cal L$. Now, for each $i\in \mathbb{Z}_{\zeta}$,  we  make a copy $\sigma_i'$ of $\sigma_i$.  We translate $\sigma_i'$ around $C$ in clockwise direction keeping the center of $\sigma_i'$ lying  on the locus $\cal L$ until $\sigma_i'$ touches the object $\sigma_{i+1}$.
 Note that  $\sigma_i'$  also touches $\sigma_i$, otherwise  $K$ is not an optimal configuration since we can place an extra object in between $\sigma_i$ and $\sigma_{i+1}$. It is easy to observe that, apart from these two objects, $\sigma_i'$ does not intersect any other objects in  ${\cal I}$.
 In this way, we obtain a set ${\cal I}'=\{{\sigma}_0'$, ${\sigma}_1'$,\ldots, ${\sigma}_{\zeta-1}'\}$ of $\zeta$ objects. The method of construction ensures that ${\cal I}'$ is an independent set (and ${\cal I}'$ together with $C$ is a standard optimal independent kissing configuration). Now, consider the clock-wise order
 $({\sigma}_0, {\sigma}_0', {\sigma}_1, {\sigma}_1',\ldots, {\sigma}_{\zeta-1}, {\sigma}_{\zeta-1}')$ of
 appearance of  objects in $\cal I\cup {\cal I}'$  around $C$. Here each object is intersected by exactly two objects: its previous and next object in the sequence. Therefore, $\cal I\cup {\cal I}'$ is a C-block of size $2\zeta$ where $C$ is the core object.
 
\end{proof}

Similar to the path of cycles, we define the  following.
\begin{definition}[Path of C-blocks]\label{def:path-C-block}
A set $\cal P$ of  C-blocks each of size $k$ from a family of geometric objects is said to form a \emph{path of C-blocks} of size $k$ and length $m$ if  the  geometric intersection graph of all the geometric objects in $\cal P$  is a path of cycles of size $k$ and length $m$.
\end{definition}

\begin{lemma}\label{lm:path-of-C-block}
For any positive integer $m$, there exists a path of C-blocks of size $2\zeta$ and length $m$ for a family of translated copies of a convex object.  

\end{lemma}

\begin{proof}
Let ${\cal B}_1, {\cal B}_2, \ldots, {\cal B}_m$ be $m$ C-blocks, as in Lemma~\ref{lm:C_block}. 
Let $H_i$ be  the convex hull  of all extreme points of objects in a C-block ${\cal B}_i$, $i\in[m]$.  
Let us fix  an $i\in [m-1]$.
By translating one of the convex-hulls around the other, we can always place two convex hulls $H_i$ and $H_{i+1}$ such that the point of intersection is an extreme point  from both the hulls and common interior of $H_i$ and $H_{i+1}$ is empty. As a result, we have an arrangement of two C-blocks ${\cal B}_i$ and ${\cal B}_{i+1}$ where exactly one object from each block intersects with the other. Let us denote the  point of intersection of ${\cal B}_i$ and ${\cal B}_{i+1}$ as $x_i$.  Let $L_i$ be a separation line  passing through $x_i$ such that ${\cal B}_i$ and ${\cal B}_{i+1}$ are in two different sides of $L_i$. 
Let $t_1$th object of ${\cal B}_i$ intersects $t_2$th object  ${\cal B}_{i+1}$  in this arrangement, where $t_1,t_2 (t_1\neq t_2) \in [2\zeta]$. If we apply the same arrangement  for all $i \in [m-1]$ such that $t_1$th object of ${\cal B}_i$ intersects $t_2$th object  ${\cal B}_{i+1}$, then  all  separation lines $\{L_i| i \in [m-1]\}$ will be parallel to each other. This ensures that 
the geometric intersection graph of these C-blocks is a path. Therefore, the claim follows.
\end{proof}

Now, we have the main result.

\begin{theorem}\label{thm:mcds_con}
Let $\zeta$ be the independent kissing number of a family $\cal S$ of translated copies of a convex object in $\IR^2$.
Then the  competitive ratio of every deterministic online algorithm  for minimum connected dominating set (MCDS) of $\cal S$ is at least $2(\zeta-1)$, if $opt_{cds}$ is one; otherwise it is at least~$\frac{2\zeta-1}{3}$.
\end{theorem}
\begin{proof}
Consider a C-block $\cal B$ of size $2\zeta$ 
 as in Lemma~\ref{lm:C_block}. Note that the geometric intersection graph of objects in $\cal B$ is a a cycle.  Consider the input sequence for  the objects in $\cal B$ in cyclone-order. Due to Lemma~\ref{lemma:path-of-cycles}, any online algorithm reports CDS of size $2\zeta -2$ for this sequence. Whereas, all the objects in   $\cal B$ can be dominated by a core object that we append as the last object in the input sequence. Therefore, competitive ratio is at least $2(\zeta-1)$, whereas  $opt_{cds}$ is one.

For general case, 
consider a path $\cal P$ of C-blocks of size $2\zeta$ and length $m$ 
 as in Lemma~\ref{lm:path-of-C-block}. Here the geometric intersection graph of the objects in $\cal P$ is a path of cycle (follows from Definition).  Consider the input sequence $S$ for  the objects in $\cal P$ in  cyclone-order. Due to Lemma~\ref{lemma:path-of-cycles}, any online algorithm reports CDS of size $(2\zeta -1)m-1$ for this sequence. Let   $\sigma \in{ \cal P}$ be the last object in the sequence $S$. In other words, $\sigma$ is the head of the last cycle in the cyclone-order.   Now, we add an object $\sigma^*$, from the same family of translated copies,  to the input sequence $S$ such  that $\sigma^*$ touches only $\sigma$. As a result, any deterministic algorithm will report at least $(2\zeta -1)$ objects as CDS from the $m$th cycle also.  At the end,  we add $m$ core objects, one from each of the C-block in $\cal P$ to the input sequence.  For this input sequence, any online algorithm reports at least $(2\zeta -1)m$ objects, whereas the size of the optimum is $3m-1$ that consists of head, tail and core objects of each C-block in $\cal P$, excepting tail from the first C-block. Hence, the result follows.  
\end{proof}

\begin{remark}
For any family of arbitrary oriented convex objects in $\IR^2$, if we have  a standard independent kissing configuration where the size of the independent set is $\zeta'$, then  we can apply the above technique to obtain a lower bound result  ($\zeta$ will be replaced by $\zeta'$) of MCDS for that family.
\end{remark}

\section{Fixed Oriented Unit Hyper-cubes in $\IR^d$}\label{3}

Throughout this section, if explicitly not mentioned,  we use hyper-cube (in short) to mean axis-parallel unit hyper-cube. 

\subsection{Independent Kissing Number}\label{3.1}


\begin{lemma}\label{cube:ub}
For a family  of fixed oriented  unit hyper-cubes  in $\IR^d$,  the independent kissing number is at most~$2^d$, where $d$ is any positive integer.
\end{lemma}
\begin{figure}[h]
  \centering
     \begin{subfigure}[b]{0.45\textwidth}
          \centering
        \includegraphics[scale=0.55]{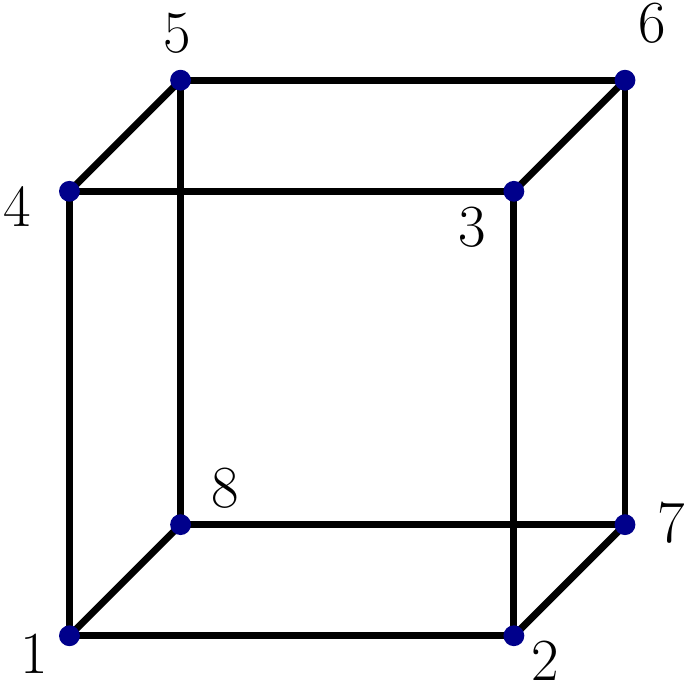}
    \caption{}
    \label{fig:cube_1}
     \end{subfigure}
     \hfill
     \begin{subfigure}[b]{0.45\textwidth}
          \centering
        \includegraphics[scale=0.6]{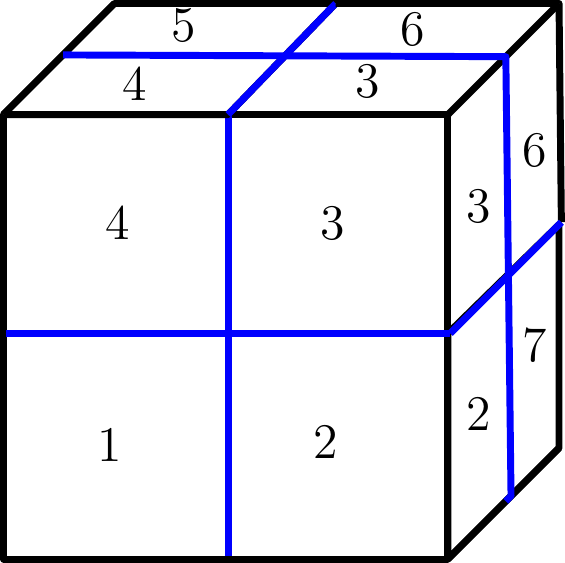}
    \caption{}
    \label{fig:cube_2}
     \end{subfigure}
       \caption{Neighbourhood of a cube in $\IR^3$.}
       \label{cube}
\end{figure}
\begin{proof}
Let $K$ be an optimal independent kissing configuration for axis parallel unit hyper-cubes in $\IR^d$. Let the core of the configuration be $u$. Consider the neighbourhood $N(u)$ that contains all the centers of hypercubes in $K\setminus \{u\}$. 
It is easy to observe that  $N(u)$ is an axis-parallel hyper-cube with side length 2 unit. Let us partition $N(u)$ into $2^{d}$ smaller symmetrical  axis-parallel hyper-cubes each having unit side (refer to Figure~\ref{cube}).
If there exist two axis-parallel hyper-cubes $u_1, u_2 \in K\setminus\{u\}$ such that their centers lie in the same unit sized smaller hyper-cube of $N(u)$, then $u_1$ and $u_2$ will overlap.
So, each smaller hyper-cube can contain at most one centre corresponding to a hyper-cube in $K\setminus \{u\}$. Since $N(u)$ has at most $2^d$ smaller hyper-cubes, $| K\setminus \{u\}|  \leq 2^d$. Therefore, the independent kissing number for axis-parallel unit hyper-cubes in $\IR^d$ is at most $2^{d}$.
\end{proof}


\begin{lemma}\label{hyper-cube}
For a family  of fixed oriented  unit hyper-cubes  in $\IR^d$, the independent kissing number is at least~$2^d$, where $d$ is any positive integer.
\end{lemma}
\begin{proof}
Here, we give an explicit construction of an independent kissing configuration $K$ where the size of the independent set is $2^d$.
 Let $\sigma_1,\sigma_2,\ldots,\sigma_{2^d}$ \& $\sigma_{2^d+1}$ be the centres of $d$-dimensional  axis parallel unit hyper-cubes of $K$. Let $\sigma_{2^d+1}=(\frac{1}{2},\frac{1}{2},\ldots,\frac{1}{2}$),  and $p_1,p_2,\ldots,p_{2^d}\in \IR^d$ be corner points of the $d$-dimensional unit hyper-cube centred at $\sigma_{2^d+1}$. It is easy to observe that each coordinate of $p_i, i\in [2^d]$ is either 0 or 1.
Let $\epsilon$ be a positive constant satisfying $0<\epsilon <\frac{1}{2\sqrt{d}}$. Let us consider
  \begin{equation}
\sigma_i(x_j) =
    \begin{cases}
      -\epsilon, & \text{ if $p_i(x_j)=0$}\\
      1+\epsilon, & \text{ if $p_i(x_j)=1$,\quad\quad\quad    for $i\in[2^d]$ \& $j\in[d]$},  
    \end{cases}      
\end{equation}
where $\sigma_i(x_j)$ and $p_i(x_j)$ are the $j^{th}$ coordinate value of $\sigma_i$ and $p_i$, respectively.

\begin{claim}
 All the hyper-cubes centred at $\sigma_1,\sigma_2,\ldots,\sigma_{2^d}$ are mutually non-touching and  each of them are intersected by the  hyper-cube  centred at $\sigma_{2^d+1}$.
 \end{claim}
 \begin{proof} 
The distance between $p_i$  and $\sigma_i$ is {$\sqrt{d}\epsilon$}, $\forall$ $i\in [2^d]$. Since $\epsilon <\frac{1}{2\sqrt{d}}$,   the corner point $p_i$ is contained in the unit  hyper-cube centred at $\sigma_i$. Thus,   $\sigma_{2^d+1}$ intersects $\sigma_i$, $\forall$ $i\in [2^d]$.
Consider, any pair $\sigma_i$ and  $\sigma_j$, for $i,j\in [2^d]$ and $i \neq j$. Since $\sigma_i$ and  $\sigma_j$
are distinct,
{the distance between $\sigma_{i}$ and $\sigma_j$ is $(1+2\epsilon)$ under $L^{\infty}$ norm.}
 This implies that all the $d$-dimensional axis-parallel unit hyper-cubes centred at $\sigma_1,\sigma_2,\ldots,\sigma_{2^d}$ are pair-wise non-touching. 
\end{proof}
This completes  the proof.
 \end{proof}

 

\begin{theorem}
 The independent kissing number $\zeta$ for geometric intersection graph of  fixed oriented $d$-dimensional hyper-cubes is $2^d$, where $d$ is any positive integer.
 \end{theorem}

\subsection{Lower Bound of MCDS}\label{3.2}
\begin{figure}[h]
\centering
\includegraphics[width=155mm]{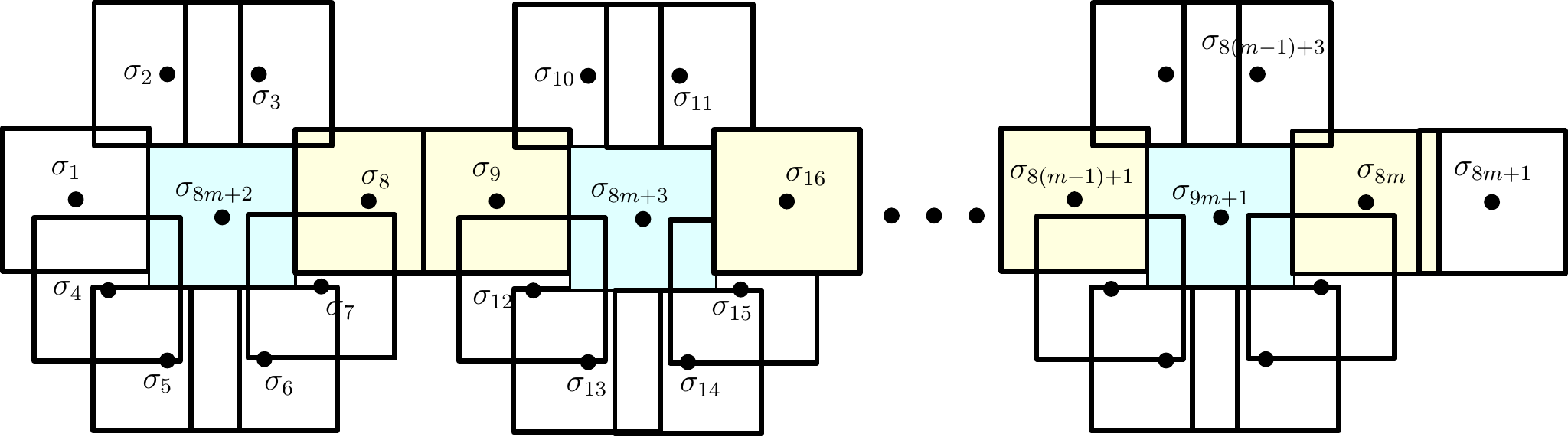}
\caption{Path of PC-blocks  for unit squares.}
\label{fig:CDS_Square}
\end{figure}
 Throughout this subsection, we use $T(d)$ to denote   $2^{d-2}$.
Now, we consider the MCDS problem.
Similar to C-block, here, we define PC-block.
\begin{definition}[PC-block]\label{def:C-block}
Let $\cal S$ be a family of geometric objects. 
A set $\cal B$ of   objects from $\cal S$ is said to form a \emph{PC-block} of length $m$ and size $k$ if following two conditions are satisfied:
\begin{itemize}
    \item the geometric intersection graph of the objects in $\cal B$ is a path  of cycles of length $m$ and size $k$, and
    \item  there exists an object $u(\notin {\cal B})$ in the family $\cal S$ that intersects all the objects in $\cal B$.
\end{itemize}  We refer the object $u$ as the \emph{core} of the  PC-block $\cal B$.
\end{definition}

\begin{lemma}\label{lm:single_group}
For a family of axis-parallel unit hyper-cubes in $\IR^d$, there exists a PC-block of  length $2^{d-2}$ and size $8$, where $d\geq 2$ is an integer.
\end{lemma}

 
   
    
    

\begin{proof}
We prove this by induction. For $d=2$, the base case of the induction follows from the construction in Figure~\ref{fig:CDS_Square}. Here, the hyper-cubes centered at $\sigma_{1}, \sigma_{2},\ldots,\sigma_{8}$  forms a PC-block of length 1 and size 8. The hyper-cube centered at $\sigma_{8m+2}$ is the core object of this PC-block. Here, the sequence $\sigma_{1}, \sigma_{2},\ldots,\sigma_{8}$ is a cyclone-order of the hyper-cubes in the PC-block.


Let us assume that the induction hypothesis holds for $d-1$ dimension. Let ${\cal B}_{d-1}$ be a PC-block of length $2^{d-3}$ and size $8$ of  axis parallel unit hyper-cubes in $\IR^{d-1}$. To construct a PC-block ${\cal B}_d$ in $\IR^d$,  we consider two copies $B_1$ and $B_2$ of ${\cal B}_{d-1}$ and append  $d$th coordinate to each of their hyper-cubes. 
 Let  $S_{d-1}:\sigma_1^{d-1}, \sigma_2^{d-1},\ldots, \sigma_{T(d-1)}^{d-1}$ be a cyclone order  of centers of unit hyper-cubes in ${\cal B}_{d-1}$. 
Now, we define the cyclone-order of center $S_d: \sigma_1^d,\sigma_2^d,\ldots,\sigma_{T(d)}^d$ for unit hyper-cubes in ${\cal B}_d$ as follows.
The first half of the sequence $S_d$ is made from the cyclone-order  $S_{d-1}$ with original order and the second half  is made from   the reverse order of $S_{d-1}$.
Now, we are going to define the $d$th coordinate. Let $0<\epsilon<\frac{1}{8}$ be a small positive constant and $f$ be any fixed real number.
For the center $\sigma_{T(d-1)}^{d-1}$ of $B_1$ and $B_2$, we append $f-(\frac{1}{2}-\epsilon)$ and $f+(\frac{1}{2}-\epsilon)$, respectively,   as the $d$th coordinate. 
For all other centers in $B_1$ and $B_2$, we append $f-(\frac{1}{2}+2\epsilon)$ and $f+(\frac{1}{2}+2\epsilon)$, respectively, as the $d$th coordinate. 
This ensures that only $\sigma_{T(d-1)}^{d-1}$ of $B_1$ and $B_2$ intersects with each other.
More formally, we have the following.
 
  \begin{equation}\label{eq:mcd_hypercube1}
\sigma_i^d(x_j) =
    \begin{cases}
      f-(\frac{1}{2}+2\epsilon), & \text{ if $j=d$ and $i< T(d-1)$}\\
      f-(\frac{1}{2}-\epsilon),& \text{ if $j=d$ and $i= T(d-1)$}\\
         f+(\frac{1}{2}-\epsilon), & \text{ if $j=d$ and $i=T(d-1)+1$}\\
     f+(\frac{1}{2}+2\epsilon), & \text{ if $j=d$ and $i> T(d-1)+1$}\\
      \sigma_i^{d-1}(x_j), & \text{if $j<d$ and $i\leq T(d-1)$}\\
      \sigma_{2 T(d-1)-i+1}^{d-1}(x_j), & \text{if $j<d$ and $i\geq T(d-1)+1$}.
        
    \end{cases}      
\end{equation}
Here $\sigma_i^d(x_j)$ is the $j^{th}$ coordinate value of $\sigma_i^{d}$,  where $ i \in [T(d)]$ and $j\in [d]$.

 This  appending of the $d$th coordinate   ensures the geometric intersection graph of the hyper-cubes in $B_i$ ($i\in [2]$) remains a path of cycles of length $2^{d-3}$ and size $8$. Additionally, the geometric intersection graph of $B_1\cup B_2$ is a path of cycles of length $2^{d-2}$ and size $8$.

Now, to prove the lemma,  we only need  to prove the existence of a core hyper-cube for ${\cal B}_d$.
Let us assume that a $(d-1)$-dimensional unit hyper-cube centred at $c^{d-1}$ dominates all hyper-cubes of ${\cal B}_{d-1}$.
The center $c^d$ of the core object is obtained by appending $f$  as the  $d$th coordinate to $c^{d-1}$. In other words, we have
\begin{equation}\label{eq:mcd_hypercube2}
c^d(x_j)=
     \begin{cases}
     f,& \text{ if $j=d$}\\
     c^{d-1}(x_j), & \text{ if $j<d$}.
     \end{cases} 
\end{equation}

Here  $c^{d}(x_j)$ is the $j^{th}$ coordinate value of  $c^{d}$, where $ i \in [T(d)]$ and $j\in [d]$.

\begin{figure}[h!]
\centering
\includegraphics[width=90mm]{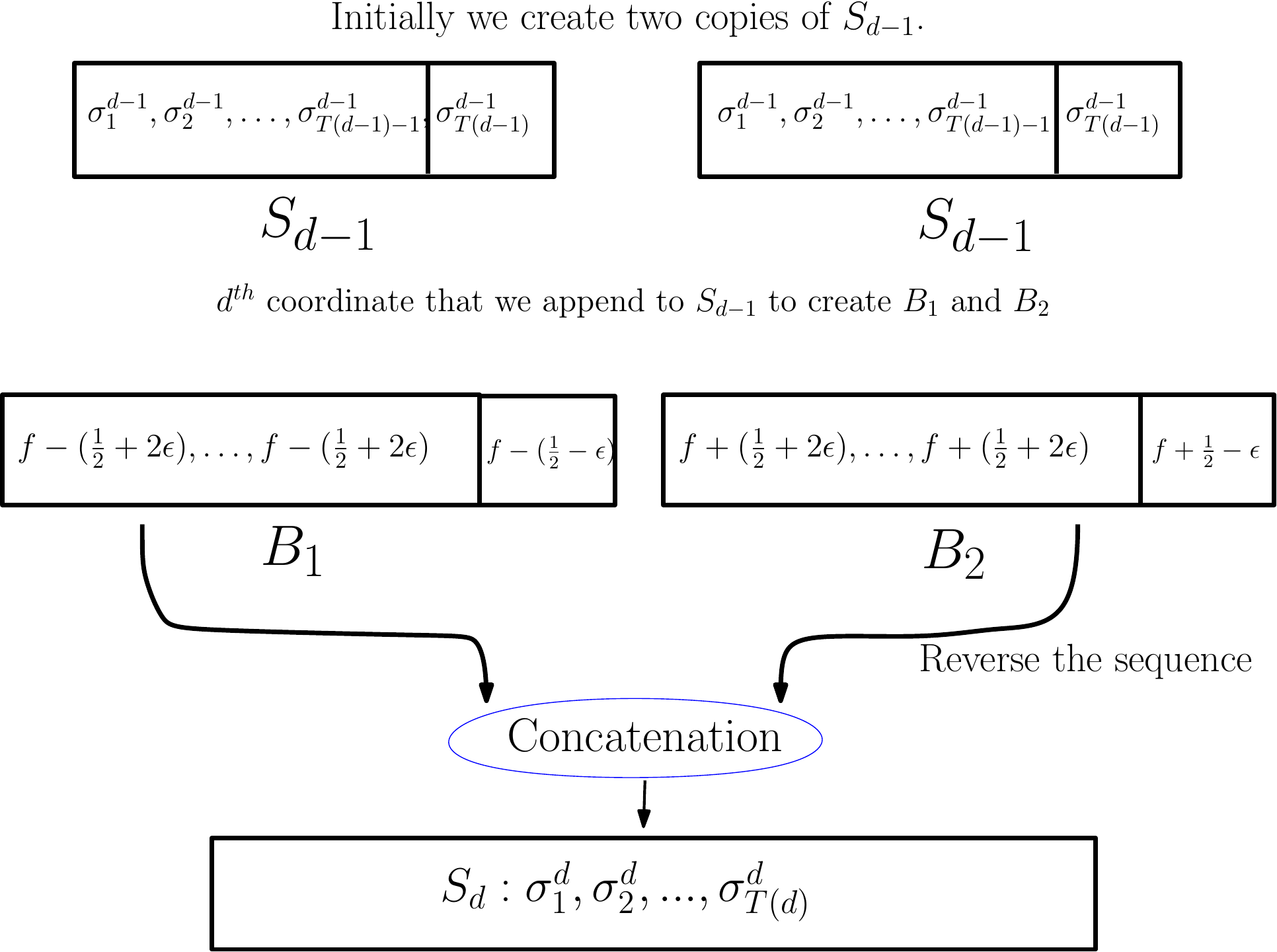}
\caption{Construction of the cyclone-order $S_d$.}
\label{Concat}
\end{figure}

 

 \begin{claim}
 The axis-parallel unit hyper-cube centered at $c^{d}$ dominates all the hyper-cubes  in ${\cal B}_{d}$.
 \end{claim}

 \begin{proof}
 To prove the claim, we need to show that $|{c^{d}(x_j)-\sigma_i^{d}(x_j)}| \leq 1$ for all $i\in [T(d)]$ and $j\in[d]$.
 Since the $(d-1)$-dimensional axis-parallel hyper-cube centred at $c^{d-1}$ dominates all the axis-parallel hyper-cubes centred in the sequence $S_{d-1}$ (due to induction hypothesis), we have $|{c^{d-1}(x_j)-\sigma_i^{d-1}(x_j)}|\leq 1$ for all $i\in [T(d-1)]$ and $j\in[d-1]$.  
From Equation~\ref{eq:mcd_hypercube1} and~\ref{eq:mcd_hypercube2}, we get  $|{c^d(x_d)-\sigma_i^d(x_d)}|\leq\frac{1}{2}+2\epsilon$. As a result, $|{c^{d}(x_j)-\sigma_i^{d}(x_j)}|\leq 1$ for all $i\in [T(d)]$ and $j\in[d-1]$. 
    Therefore, the claim holds.
 \end{proof}
 
 This completes the proof.
 \end{proof}

\begin{definition}[Path of PC-blocks]\label{def:path-PC-block}
A set $\cal P$ of  PC-blocks from a family of geometric objects is said to form a \emph{path} if  following conditions are satisfied:
\begin{itemize}
    \item[(a)] the geometric intersection graph of these PC-blocks is a path, and 
   \item[(b)]  if two PC-blocks are intersecting, then exactly one object from each block intersects with the other.
\end{itemize}
The \emph{length} of a path $\cal P$ of  PC-blocks  is defined as the  number of PC-blocks in $\cal P$.
\end{definition}

Now, we prove the following.

\begin{lemma}\label{lm:path_of_PC_block}
For a family of axis-parallel unit hyper-cubes in $\IR^d$, there exists a path of PC-block of length $m$ where each PC-block is of length $2^{d-2}$ and size $8$, where $m \geq 1, d\geq 2$ are integers.
\end{lemma}

\begin{proof}
For $d=2$,  the example given in Figure~\ref{fig:CDS_Square}  is a path of PC-blocks satisfying the lemma. For $d>2$, we give an explicit construction.

First, we define  each PC-block. Consider a PC-block ${\cal B}_d$ as defined in Lemma~\ref{lm:single_group}. Let $S_d: \sigma_1^{d}, \sigma_2^{d},\ldots, \sigma_{T(d)}^{d}$ be a cyclone-order of the centers of hyper-cubes in ${\cal B}_d$. Let $c^d$ be the center of the core for the PC-block ${\cal B}_d$ with the value of the $d$th coordinate $f$.  Let us  change the value 
of the $d$th coordinate of $\sigma_1^{d}$ and $\sigma_{T(d)}^{d}$ as $f-(1-\epsilon)$ and $f+(1-\epsilon)$, respectively. Note that, after changing this, the modified  hyper-cubes in  ${\cal B}_d$ remains a PC-block. We denote this modified PC-block as ${\cal B}(f)$. 

Now, we define a path of PC-blocks where each PC-block is a translated copy of ${\cal B}(f)$.  Consider translated copies ${\cal B}(f+t(3-2\epsilon))$, where $1\leq t \leq m-1$.  Note that only $\sigma_1^{d}$ of ${\cal B}(f+t(3-2\epsilon))$  intersects $\sigma_{T(d)}^{d}$ of ${\cal B}(f+(t-1)(3-2\epsilon))$. As a result, ${\cal B}(f)$, ${\cal B}(f+(3-2\epsilon)), \ldots$, ${\cal B}(f+(m-1)(3-2\epsilon))$ forms a path of PC-block  of length $m$. Thus the lemma follows.

\end{proof}

\begin{theorem} \label{lm:MCDS_hyper_cube_lb}
The competitive ratio of every deterministic online algorithm for minimum connected dominating set (MCDS) of $d$-dimensional unit hyper-cubes in the online setup is at least $7\times 2^{d-2}-1$, if $opt_{cds}=1$; otherwise it is at least~$\frac{7\times 2^{d-2}}{3}$.
\end{theorem}
\begin{proof}
Similar to the proof of Theorem~\ref{thm:mcds_con}, we can prove this theorem  with the following refinements.

\begin{itemize}
    \item Instead of C-block of size $2\zeta$, here we have PC-block of size $8$ and length $2^{d-2}$ (due to Lemma~\ref{lm:single_group}).
    \item Instead of path of C-block, we have path of PC-block (due to Lemma~\ref{lm:path_of_PC_block}).
   
\end{itemize}
\end{proof}

 \section{Arbitrary Oriented Unit Hyper-cubes in $\IR^d$}\label{4}
 \subsection{Independent Kissing Number}\label{3.2}
 
  \begin{figure}[h]
     \begin{subfigure}[b]{0.45\textwidth}
\centering
\includegraphics[scale=0.34]{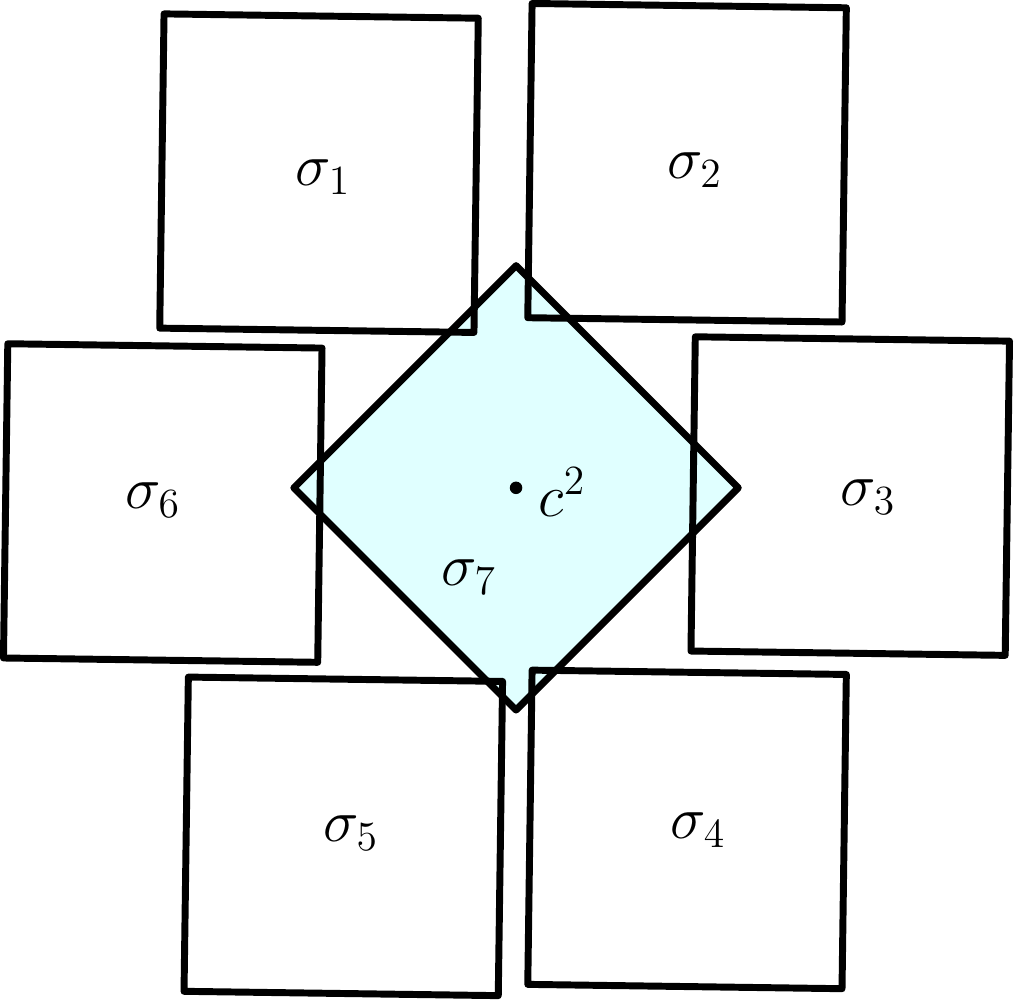}
\caption{}
\label{fig:arb_sq}
 \end{subfigure}
  \begin{subfigure}[b]{0.45\textwidth}
\centering
\includegraphics[scale=0.5]{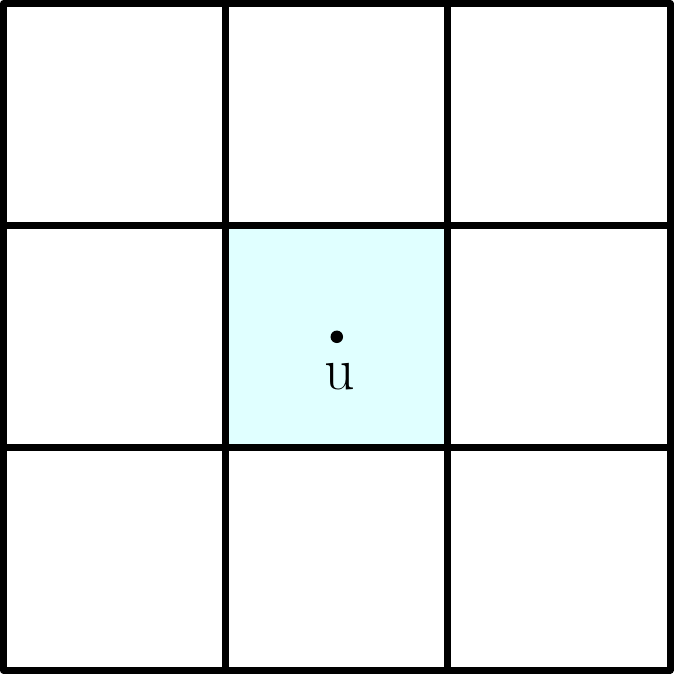}
\caption{}
\label{fig:kiss_sq}
\end{subfigure}
\label{}
\caption{(a) An independent kissing configuration for arbitrary oriented unit squares, (b) The optimal kissing  configuration for squares.}
  \end{figure}
First, we consider the arbitrary oriented unit squares in $\IR^2$.
Youngs~\cite{Youngs} proved that  the kissing number for squares is 8. Later Klamkin~\cite{KlamkinLL} proved that the Figure~\ref{fig:kiss_sq} is the unique optimal configuration for this. Since in this configuration, the neighbours are touching, the independent kissing number $\zeta$ for family of arbitrary oriented unit squares is strictly less than eight. In Figure~\ref{fig:arb_sq}, we give an example of independent kissing configuration for arbitrary oriented unit squares where  size of the independent set is six.   Therefore, we raise the following question.
 
 
 \begin{problem}
 What is the independent kissing number  for geometric intersection graph of  arbitrary oriented squares: six or seven?
 \end{problem}


Now, we give a lower bound for the general case.
 \begin{lemma}\label{lm:single_group_arb}
 The  value of independent kissing number $\zeta$ for geometric intersection graph of  arbitrary oriented $d$-dimensional unit hyper-cubes is at least~$6 \times 2^{d-2}$, where $d\geq 2$ is an integer.
 \end{lemma}
 \begin{proof}
 We prove this by induction on $d$.
For $d=2$, the base case of the induction,  consider the   independent kissing configuration given in Figure~\ref{fig:arb_sq} where size of the independent set is 6.


Let us assume that the induction hypothesis holds for $d-1$ dimension. Let  $K_{d-1}$ be an independent kissing configuration where size of the independent set is $M(d-1)=6\times 2^{d-3}$ for $(d-1)$-dimensional unit hyper-cubes. Let the core hyper-cube in $K_{d-1}$ is centered at $c^{d-1}$, and other  hyper-cubes  are centred at $\sigma_1^{d-1}, \sigma_2^{d-1},\ldots, \sigma_{M(d-1)}^{d-1}$. To construct independent kissing configuration $K_d$, we take two copies independent set from $K_{d-1}$ and append $d$th coordinate to each of them. We append  $-(\frac{1}{2}+\epsilon)$ as the   $d$th coordinate to all the centers of the first copy, whereas we append $-(\frac{1}{2}+\epsilon)$ to the other. More formally, 

\begin{equation}\label{eq:mcd_arb_hypercube1}
\sigma_i^d(x_j) =
    \begin{cases}
      \sigma_i^{d-1}(x_j), & \text{ if $j<d$ and $i\leq M(d)$}\\
      \frac{1}{2}+\epsilon,& \text{ if $j=d$ and $i\leq M(d-1)$}\\
      -(\frac{1}{2}+\epsilon), & \text{ if $j=d$ and $i > M(d-1)$.}
    \end{cases}      
\end{equation}
The center $c^d$ is obtained by appending $0$ as the  $d$th coordinate to $c^{d-1}$, and the hyper-cube centered at $c^d$ is axis parallel in every dimension except first two coordinates. Therefore, we have.
\begin{equation}\label{eq:mcd_arb_hypercube2}
c^d(x_j)=
     \begin{cases}
     $0$,& \text{ if $j=d$}\\
     c^{d-1}(x_j), & \text{ if $j<d$.}
     \end{cases} 
\end{equation}
Here, all the $d$-dimensional unit hyper-cubes centred at $\sigma_1^d,\sigma_2^d,\ldots,\sigma_{M(d)}^d$ are mutually non-touching and all of them can be dominated by the unit hyper cube centred at $c^{d}$. It is straightforward to see that the value of $M(d)= 2M(d-1)$, where $M(d-1)=6\times 2^{d-3}$. 
 Hence the lemma follows.
 \end{proof}

Now, we raise following question.
 \begin{problem}
 What is the independent kissing number $\zeta$  for geometric intersection graph of  arbitrary oriented $d$-dimensional unit hyper-cubes?
 \end{problem}
 

\subsection{Lower Bound of MCDS}

\begin{figure}[h]
\centering
\includegraphics[width=155mm]{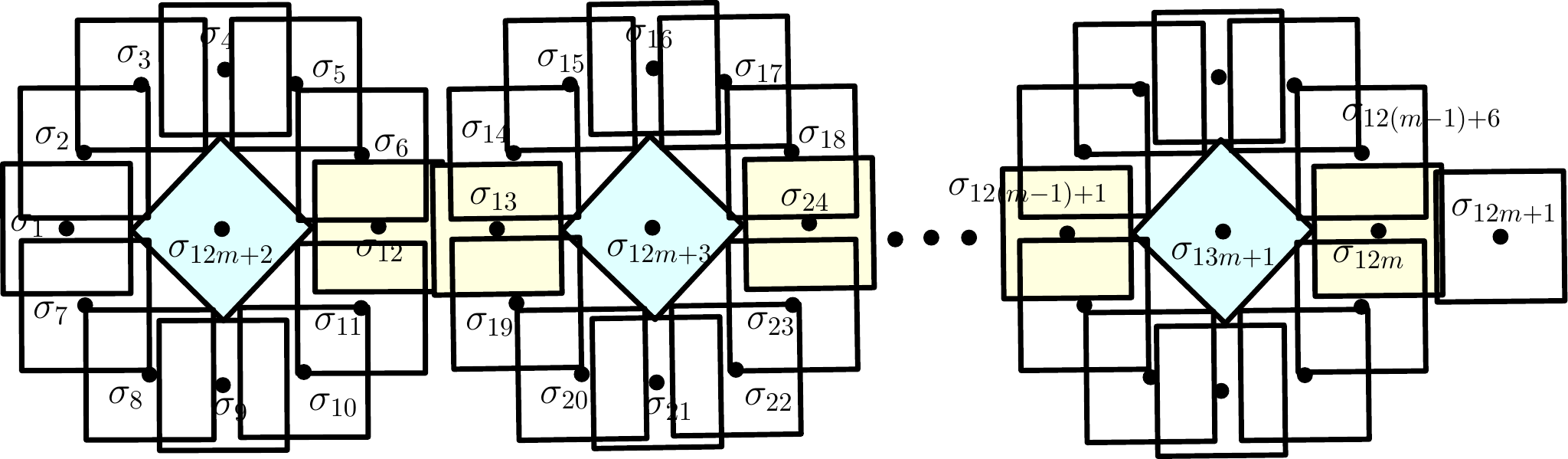}
\caption{Path of PC-blocks for arbitrary oriented unit squares.}
\label{fig:arb_sq1}
\end{figure}
 

 \begin{lemma}\label{lm:single_group_arb}
 For a family of arbitrary oriented unit hyper-cubes in $\IR^d$, there exists a PC-block of  length $2^{d-2}$ and size $12$, where $d\geq 2$ is an integer.
 \end{lemma}

 
   
    

\begin{proof}
The proof of this lemma is similar to Lemma~\ref{lm:single_group}.
For $d$=2, the base of the induction is followed from  Figure~\ref{fig:arb_sq1} where  we have  PC-block of length 1 and size 12. Note that the core hyper-cube of the PC-block is arbitrary oriented only.

The PC-block in higher dimension is constructed in a similar way to Lemma~\ref{lm:single_group} where all the hyper-cubes are axis-parallel. 
Similar way, we define the core hyper-cube of the PC-block in higher dimension. The only difference is that this core hyper-cube is also axis-parallel, except the first two coordinates.  
\end{proof}

Similar to Lemma~\ref{lm:path_of_PC_block}, using Figure~\ref{fig:arb_sq1} instead of Figure~\ref{fig:CDS_Square},  we can prove the following.

\begin{lemma}\label{lm:path_of_PC_block_arbit}
For a family of arbitrary oriented unit hyper-cubes in $\IR^d$, there exists a path of PC-block of length $m$ where each PC-block is of length $2^{d-2}$ and size $12$, where $m \geq 1, d\geq 2$ are integers.
\end{lemma}

Now, similar to Theorem~\ref{lm:MCDS_hyper_cube_lb},  we have the following lower bound result.

\begin{theorem} \label{lm:MCDS_arb_hyper_cube_lb}
The competitive ratio of every deterministic online algorithm for minimum connected dominating set (MCDS) of $d$-dimensional arbitrary oriented unit hyper-cubes in the online setup is at least $11\times 2^{d-2}-1$, if $opt_{cds}=1$; otherwise it is at least~$\frac{11\times 2^{d-2}}{3}$.
\end{theorem}


\section{Congruent Balls in $\IR^3$}\label{5}

\subsection{Independent Kissing Number}\label{4.1}

  


First, we are going to present a lower bound.

\begin{lemma}\label{lm:mds_ball}
The independent kissing number  for geometric intersection graph of  unit balls in $\IR^3$ is at least~12. 
\end{lemma}


\begin{proof}
Let $\epsilon >0$ (in particular, one can choose $\epsilon$=$0.001$), and $e$=$2+\epsilon$ be the length of each edge of a regular icosahedron $C$. Choose corner points of $C$ as the centre $\sigma_1,\sigma_2,\ldots,\sigma_{12}$ of the  unit balls (having unit radius). 
Since the edge length of the icosahedron $C$ is greater than 2, so all of these balls are non-touching.  It is a well-known fact that if a regular  icosahedron has edge length $e$, then the radius $R$ of the circumscribed ball is $R$=$e \sin(\frac{2\pi}{5})$. For our case, it is easy to see that $R<2$ for the icosahedron $C$.
Thus all the unit balls centred at $\sigma_1,\sigma_2,\ldots,\sigma_{12}$ can be dominated by a unit ball whose center $\sigma_{13}$ coincides with the centre of the icosahedron $C$. This implies that the lower bound of $\zeta$ for unit balls is 12.
\end{proof}

 Since the kissing number for balls in $\IR^3$ is 12~\cite{BrassMP,Schutte},  we have the following.

\begin{theorem}
 The independent kissing number  for geometric intersection graph of  unit balls in $\IR^3$ is~12.
 \end{theorem}

\subsection{Lower Bound of MCDS}\label{4.2}
\begin{theorem}\label{thm:mcds_ball}
The  competitive ratio of every deterministic online algorithm  for minimum connected dominating set (MCDS) of unit balls is at least 17, if $opt_{cds}$ is one; otherwise it is at least~$6$.
\end{theorem}
\begin{figure}[h]
\centering
\includegraphics[scale=0.35]{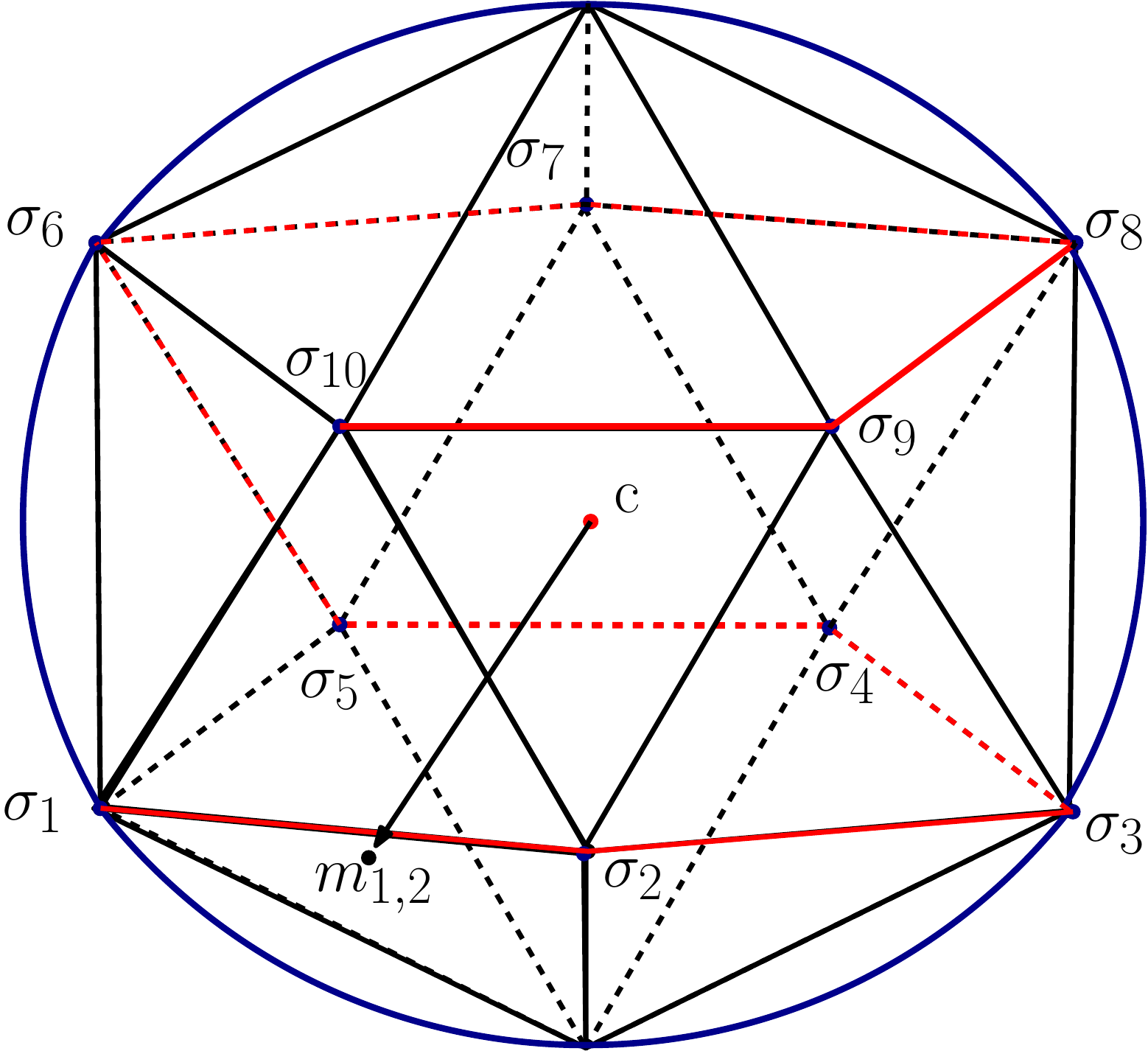}
\caption{Regular Icosahedron.}
\label{fig:icosa}
\end{figure}

\begin{proof}
Refer to Figure~\ref{fig:icosa}. Let the given regular icosahedron  has edge length same as in Lemma~\ref{lm:mds_ball}. Recall that in this icosahedron, all the unit balls centred at the vertices are pair-wise non-touching and all of them are intersected by the unit ball $\rho$ centred at the center $c$ of the icosahedron.  Consider any edge $a,b$  of the icosahedron.  Let $m_{a,b}$ be a point obtained by  shooting a ray from $c$ to the midpoint of $ab$ and extending it upto distance $R$ from $c$. If we place a unit ball centred at $m_{a,b}$ then it is intersected by  the unit ball $\rho$.  Consider the set of corners $\sigma_1,\sigma_2,\ldots,\sigma_{10}$. In a similar fashion, we can place balls in each edge of $i,i+1$ of the icosahedron, where $i\in[9]$. Note that in this way, we have 19  balls whose geometric intersection graph forms a path and each of these balls are dominated by the unit ball $\rho$. We refer these set of 19 balls as a \emph{P-block} and the unit ball $\rho$ as the core of this P-block.   Now consider the sequence of unit balls centred at $\sigma_1,m_{1,2},\sigma_2,m_{2,3}, \ldots, \sigma_9,m_{9,10}, \sigma_{10}$ that appears according to the path. For this sequence of input, any online algorithm requires at least $17$ unit balls to form a connected dominating set. Therefore, if $opt_{cds}=1$, then the   lower bound is $17$. For the general case, we can form a path of P-blocks of length $m>2$ (similar to Definition~\ref{def:path-PC-block}). Consider  the enumeration of unit balls according to the order they appear in the path of P-blocks followed by the $m$ core balls. For this input sequence,  any online algorithm will report CDS of size at least $19m-2$, whereas the size of the optimum is $3m-1$.  
Thus, the theorem follows.  
\end{proof}

\section{Fixed Oriented Unit Triangles}\label{6}
In this section, we are going to refer an equilateral triangle as an unit triangle. Observe that the neighbourhood of an unit triangle would be a hexagon as shown in Figure~\ref{fig:nbd_tri}.
 \subsection{Independent Kissing Number}\label{5.1}

\begin{figure}[h]
  \centering
     \begin{subfigure}[b]{0.30\textwidth}
          \centering
        \includegraphics[scale=0.49]{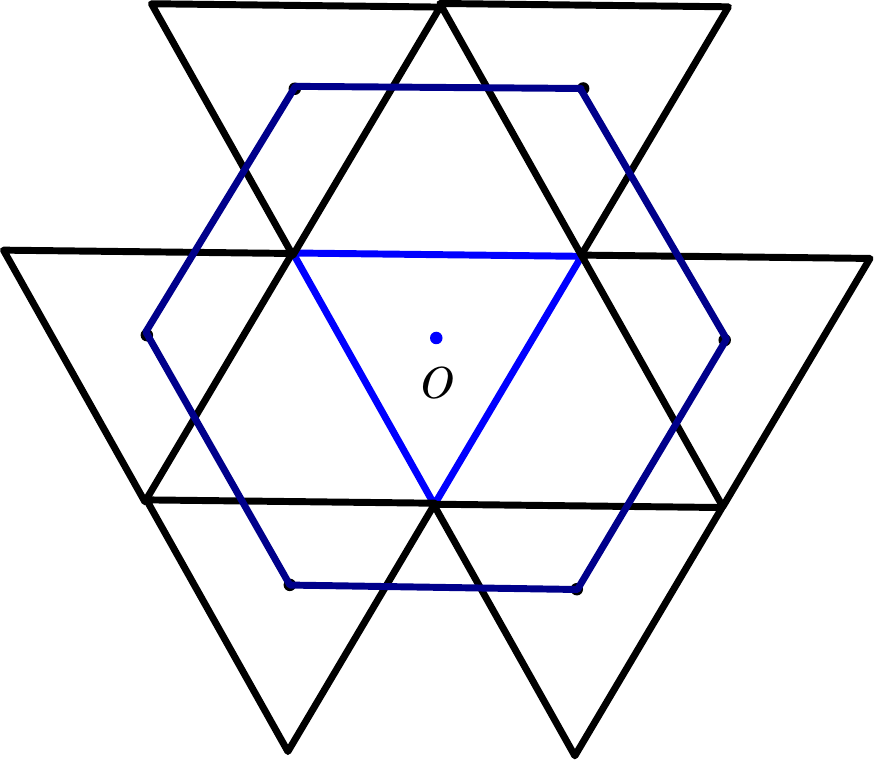}
    \caption{}
    \label{fig:nbd_tri}
     \end{subfigure}
     \hfill
     \begin{subfigure}[b]{0.30\textwidth}
          \centering
        \includegraphics[scale=0.8]{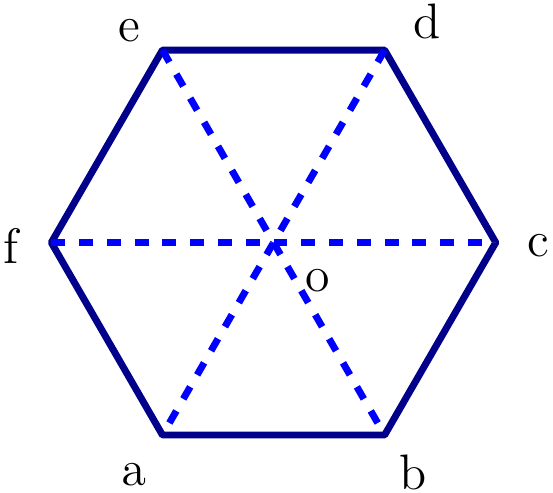}
    \caption{}
    \label{fig:tri_nbd}
     \end{subfigure}
       \hfill
     \begin{subfigure}[b]{0.30\textwidth}
          \centering
        \includegraphics[scale=0.8]{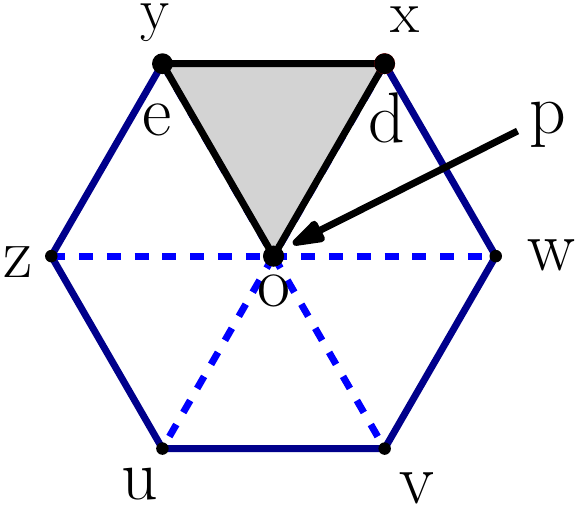}
    \caption{}
    \label{fig:tri_ode}
     \end{subfigure}
     
       \caption{(a) The boundary of the neighbourhood of a triangle (in blue) is marked blue, (b) The neighbourhood $N(\sigma_o)$ partitioned into 6 symmetrical triangles, (c) The point $p$ coincides with the point $o$. }
\end{figure}
 
 \begin{figure}[h]
  \centering
     \begin{subfigure}[b]{0.45\textwidth}
           \centering
        \includegraphics[scale=0.75]{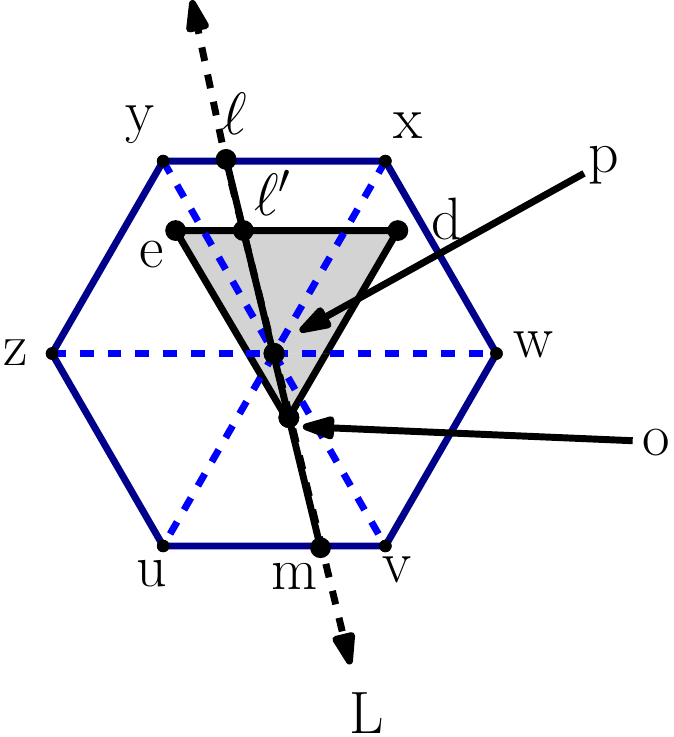}
    \caption{}
     \label{fig:tri_case1}
     \end{subfigure}
     \hfill
     \begin{subfigure}[b]{0.45\textwidth}
          \centering
        \includegraphics[scale=0.75]{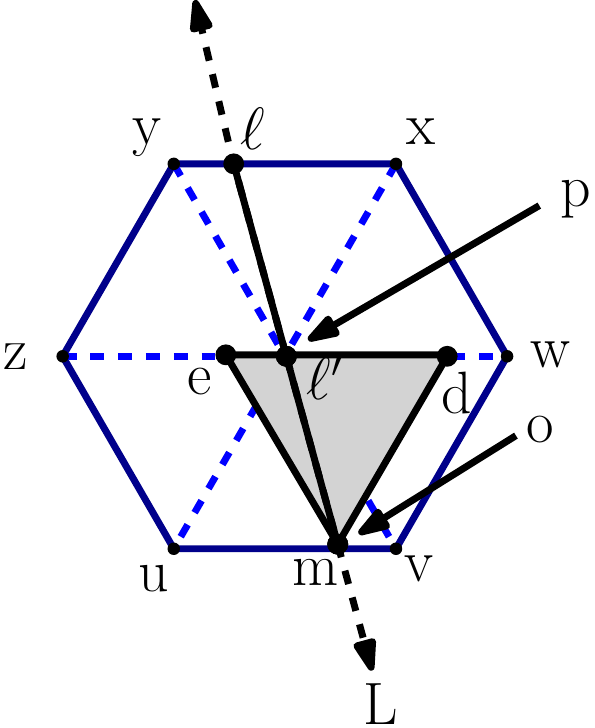}
    \caption{}
    \label{fig:tri_case2}
     \end{subfigure}
     \caption{(a) $\ell,{\ell}'$ and $m$ be the points of intersection of the line $L$ with line segments $xy, de$ and $uv$, respectively, (b) The point $p$ lie inside the $\Delta ode$.}
\end{figure}

\begin{lemma} \label{tri:ub}
The value of independent kissing number $\zeta$ for geometric intersection graph of  fixed oriented unit triangles is at most $6$.
\end{lemma}
\begin{proof}
Let $K$ be an optimal independent kissing configuration for  unit triangles. Let the core of the configuration be the triangle $\sigma_o$ centered at $o$. Consider the neighbourhood $N(\sigma_o)$ that contains all the centers of triangles in $K\setminus \{\sigma_o\}$.

Let us partition the neighbourhood $N(\sigma_o)$, denoted by the hexagon $abcdef$, into six sub-regions each consisting of an equilateral triangle with unit side as shown in the Figure~\ref{fig:tri_nbd}.
We will show that each sub-region can contain at most one triangle that belongs to $K\setminus \{\sigma_o\}$.
Now, we are going to prove this for the triangular sub-region $\Delta ode$ (proof for all other triangular sub-regions will be similar).
To prove this, it is enough to show that if we put an equilateral triangle $\sigma_p \in K\setminus \{\sigma\}$  centred at any point $p\in \Delta ode$, $N(\sigma_p)$ will cover the entire triangular sub-region $\Delta ode$.
We denote the hexagon $N(\sigma_p)$ by $uvwxyz$.

\begin{itemize}
\item[(I)] The point $p$ coincides with the point {$o$}. 
 It is easy to observe that $N(\sigma_p)$ is entirely covering \( \Delta ode \) (see Figure~\ref{fig:tri_ode}).

\item[(II)]  The point $p$ lies on the line $od$ or $oe$.
Without loss of generality, assume that $p$ lies on the line $od$. One way to view this is as follows: initially, as in case I, the point $p$ coincides with the point $o$; next, we move the point $p$ along the line $od$ towards the point $d$ (until it reaches its new destination). We can alternatively visualize this as follows: instead of moving the point $p$ on $od$ towards $d$, fixing the point $p$ and its neighbourhood $N(\sigma_p)$, we are moving the \( \Delta ode \) from its initial position  such that the point $o$ always lies on the line $pu$ and $o$ moves towards $u$.
Since  $pu$ and $yz$ are parallel and are of same length,
 $e$ will move along the line $yz$ towards $z$ and when $o$ reaches $u$ the \( \Delta ode \) is moved to the position of \( \Delta upz \).
Hence $N(\sigma_p)$ completely covers the \( \Delta ode \).

\item[(III)] The point $p$ lies inside the \( \Delta ode \). 
Let $L$ be the line obtained by extending the line segment $op$ in both direction. Let $\ell,{\ell}'$ and $m$ be the points of intersection of the line $L$ with line segments $xy, de$ and $uv$, respectively ({see Figure}~\ref{fig:tri_case1}).  One way to visualize the situation is as follows: initially, as in case I, the point $p$ coincides with the point $o$; next, we move the point $p$ along the line $L$ towards the point ${\ell}'$ (until it reaches its new destination). We can alternatively visualize this as follows: instead of moving $p$, we fix the the point $p$ and its neighbourhood $N(\sigma_p)$;  we move the $\Delta ode$ from its initial position $\Delta pxy$ such that the point $o$ moves along the line $L$ towards the point $m$. Observe that, when
$o$ reaches the boundary $uv$, the point
$p$ coincides with the point ${\ell}'$ (see Figure~\ref{fig:tri_case2}). Thus the \( \Delta ode \) is always contained in the neighbourhood $N(\sigma_p)$. This completes the proof.
\end{itemize} 
\end{proof}




\begin{figure}[h]
\centering
\includegraphics[width=40mm]{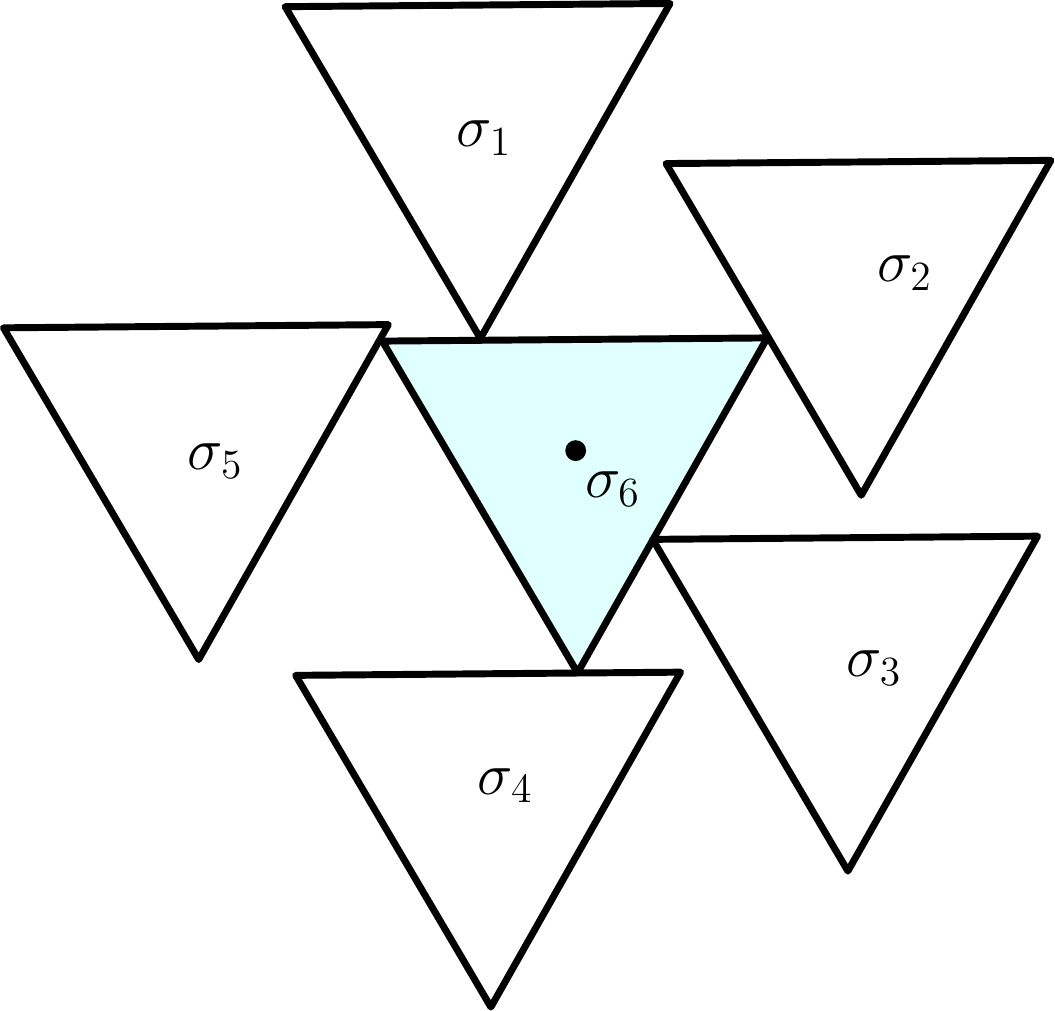}
\caption{An independent kissing configuration for unit triangles.}
\label{fig:triangle}
\end{figure}



In Figure~\ref{fig:triangle}, we have constructed an example of independent kissing configuration for unit triangles where size of the independent set is five. Thus, we raise the following question.

\begin{problem}
What is the   independent kissing number  for geometric intersection graph of  fixed oriented unit triangles- five or six?
\end{problem}

\begin{figure}[h]
\centering
\includegraphics[width=155mm]{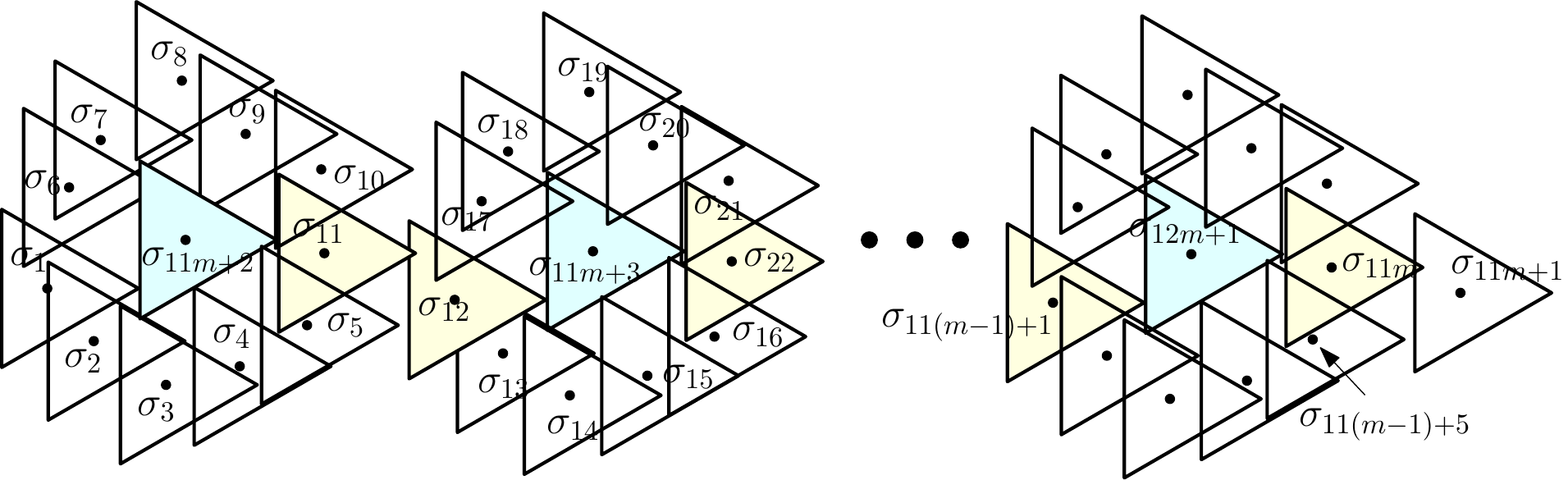}
\caption{Input instance of triangles for connected dominating set.}
\label{fig:CDS_triangle}
\end{figure}

\subsection{Lower Bound of MCDS}\label{5.2}
Since the lower bound of independent kissing number for translates of unit triangle is 5, using Theorem~\ref{thm:mcds_con} we have the lower bound of MCDS problem for unit triangles is at least 8 if $opt_{cds}$ is 1; otherwise it is 3. Here, we have constructed an example (Figure~\ref{fig:CDS_triangle}) of a path of C-blocks for unit triangles of length $m$ and size 11. Thus, similar to the proof of Theorem~\ref{thm:mcds_con}, we can prove  the following.
\begin{theorem}
The  competitive ratio of every deterministic online algorithm  for minimum connected dominating set (MCDS) of triangles is at least $9$, if $opt_{cds}$ is one; otherwise it is at least~$\frac{10}{3}$.
\end{theorem}

\section{Unit Regular $k$-gon ($k \geq 5$)}\label{7}
  
  In this section, we consider translated copies of a regular polygon as a set of input to the online algorithm. 
  
   \subsection{Fixed Oriented }\label{6.1}
\begin{figure}[h]
  \centering
     \begin{subfigure}[b]{0.32\textwidth}
         \centering
\includegraphics[scale=0.35]{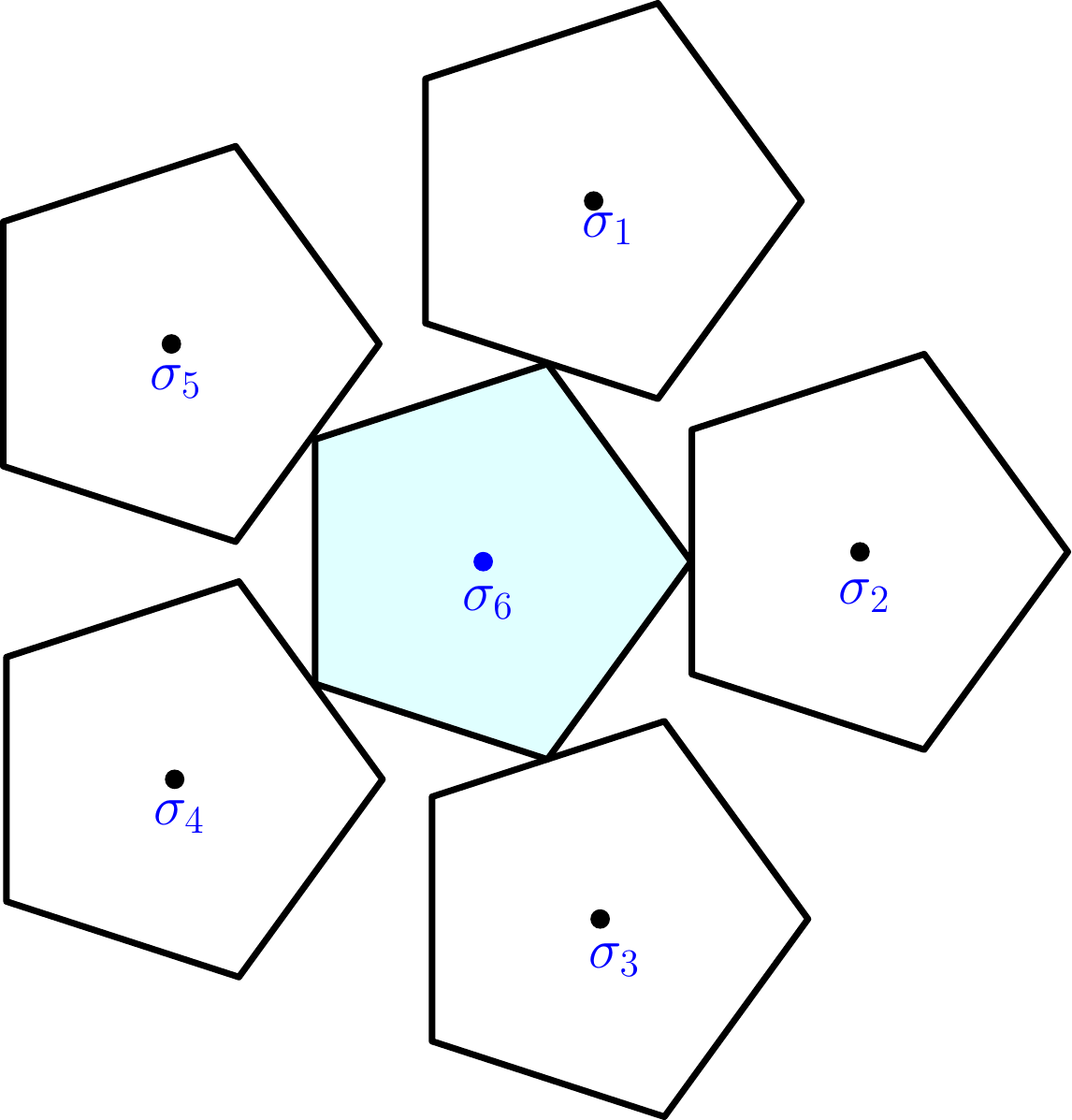} 
\caption{}
\label{fig:pentagon}
     \end{subfigure}
     \hfill
     \begin{subfigure}[b]{0.32\textwidth}
         \centering
\includegraphics[scale=0.35]{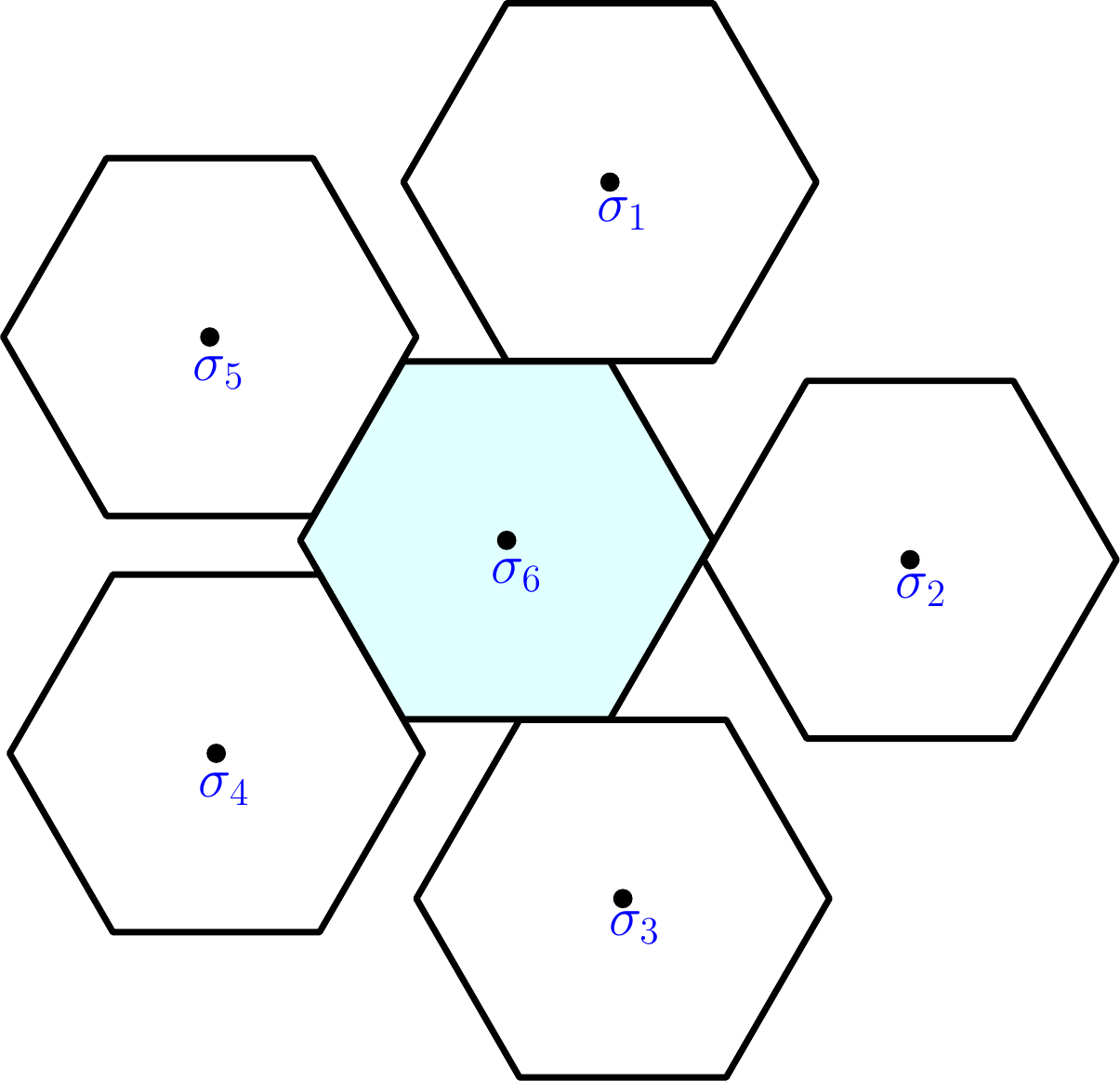}
\caption{}
\label{fig:hexagon}
     \end{subfigure}
     \begin{subfigure}[b]{0.32\textwidth}
         \centering
\includegraphics[scale=0.35]{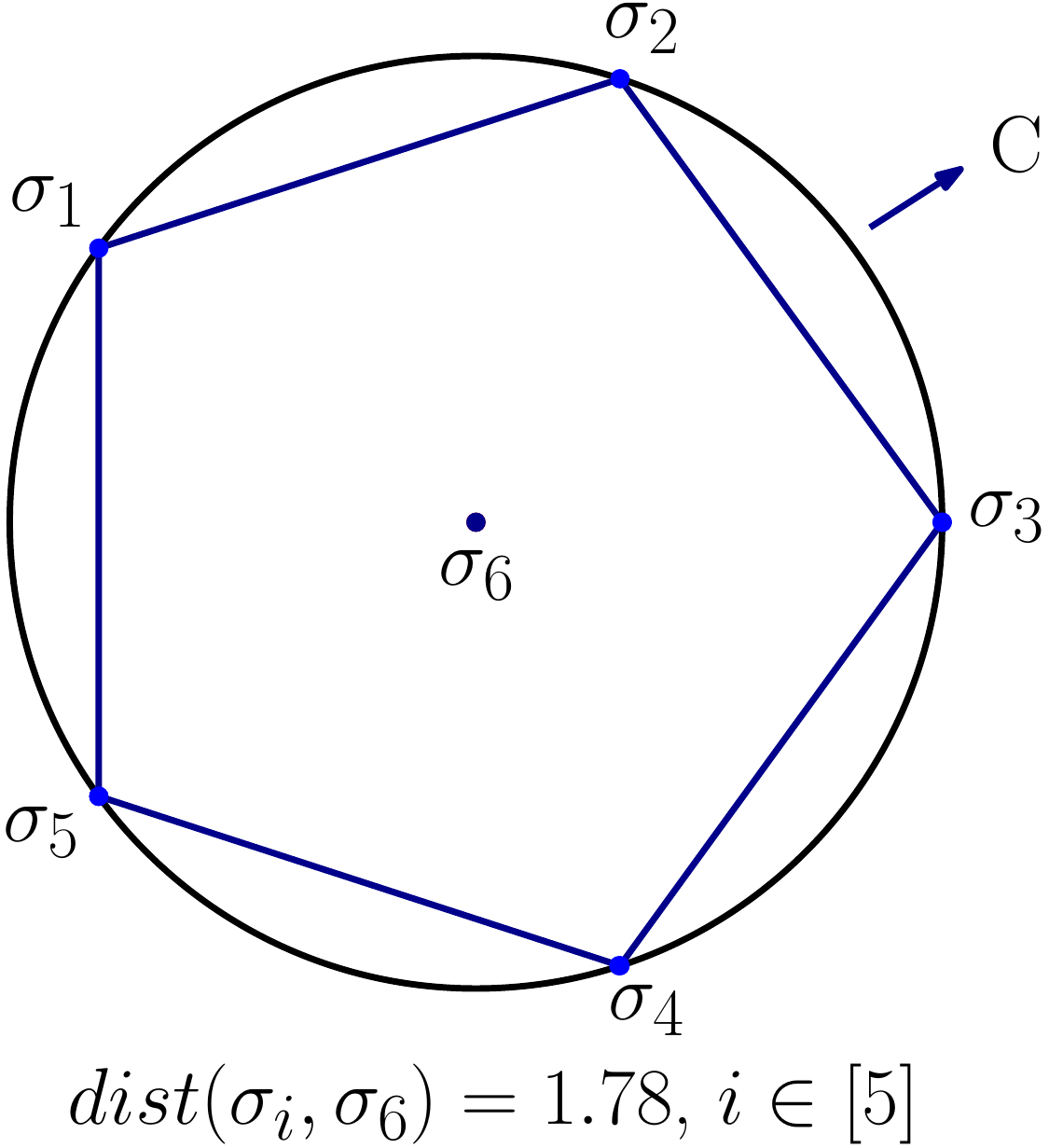}
\caption{}
\label{fig:regular}
     \end{subfigure}
       \caption{Illustration of Lemma~\ref{lm:k_7}: An independent kissing configuration for (a)  pentagon, (b)  hexagon, and (c) $k$-gon, $k \geq 7$, respectively.}
       \label{fig:penta_hexa}
\end{figure}


\begin{lemma}\label{lm:k_7}
 The value of independent kissing number $\zeta$ for geometric intersection graph of $k$-regular polygon is at least $5$, where $k\geq 5$.
\end{lemma}

\begin{proof}
For $k=5$ and $6$, refer to  Figure~\ref{fig:penta_hexa}, where we have constructed examples of independent kissing configuration having the  independent set of size five. For $k\geq 7$, we have the following proof.
Consider a circle $C$ of radius $1.78$ centred at $\sigma_6$ (see Figure~\ref{fig:regular}). Let us inscribe a pentagon of largest radius inside $C$. Let $\sigma_1,\sigma_2,\ldots,\sigma_5$ be the five corners of this pentagon. Observe that $dist(\sigma_i,\sigma_j)= 2\times 1.78 \sin(\pi/5)  > 2$, where $i,j\in[5]$ and $i\neq j$. Therefore, all the regular $k$-gons (having circum-radius $1$) centred at $\sigma_1, \sigma_2,\ldots, \sigma_5$ are pair-wise non-touching.

Let $R$ and $r$ be the circum-radius and in-circle radius of a regular $k$-gon, respectively. It is easy to observe that the relation between $R$ and $r$ is $r=R\cos({\frac{\pi}{k}})$. Clearly, $r$ is the minimum distance from the center to any edge of regular $k$-gon. For $k \geq 7$, the minimum distance from the center to the edge of a regular $k$-gon (having circum-radius $1$) is $\cos{\frac{\pi}{k}}$ which is at least $\cos{\frac{\pi}{7}}$ i.e., approx~0.9.
 Since, the  $dist(\sigma_i,\sigma_6)$ is less than $1.8$ for all $i\in[5]$,  the regular $k$-gon centred at $\sigma_6$ intersects all the regular $k$-gon centred at $\sigma_1, \sigma_2,\ldots, \sigma_5$.  Hence, 5 is the lower bound.
\end{proof}

On the other hand, the kissing number for regular $k$-gon 
 ($k \geq 5$) is 6~\cite{BrassMP,Schutte}.  Thus, we have the following question.

 
 \begin{problem}
 What is the independent kissing number  for geometric intersection graph of fixed oriented $k$-gon ($k \geq 5$): five or six?
 \end{problem}




  \subsection{Arbitrary Oriented}\label{8.1}

From Figure~\ref{arb}, we can observe that the independent kissing number for $k$-regular polygon ($5\leq k\leq10$) is at least 6. Zhao and Xu~\cite{ZHAO1998293} proved that the kissing number of $k$-regular polygon ($k \geq 5)$ is 6. As a result, the independent kissing number for $k$-regular polygon ($5\leq k\leq10$) is exactly 6. 
Since independent kissing number of unit disks is $5$. Therefore, we raise the following question:
\begin{problem}
 What is the minimum value of $k$ for which the geometric intersection graph of arbitrary oriented $k$-gon has  the independent kissing number 5?
 \end{problem}

\begin{figure}[h]
  \centering
     \begin{subfigure}[b]{0.32\textwidth}
         \centering
\includegraphics[scale=0.3]{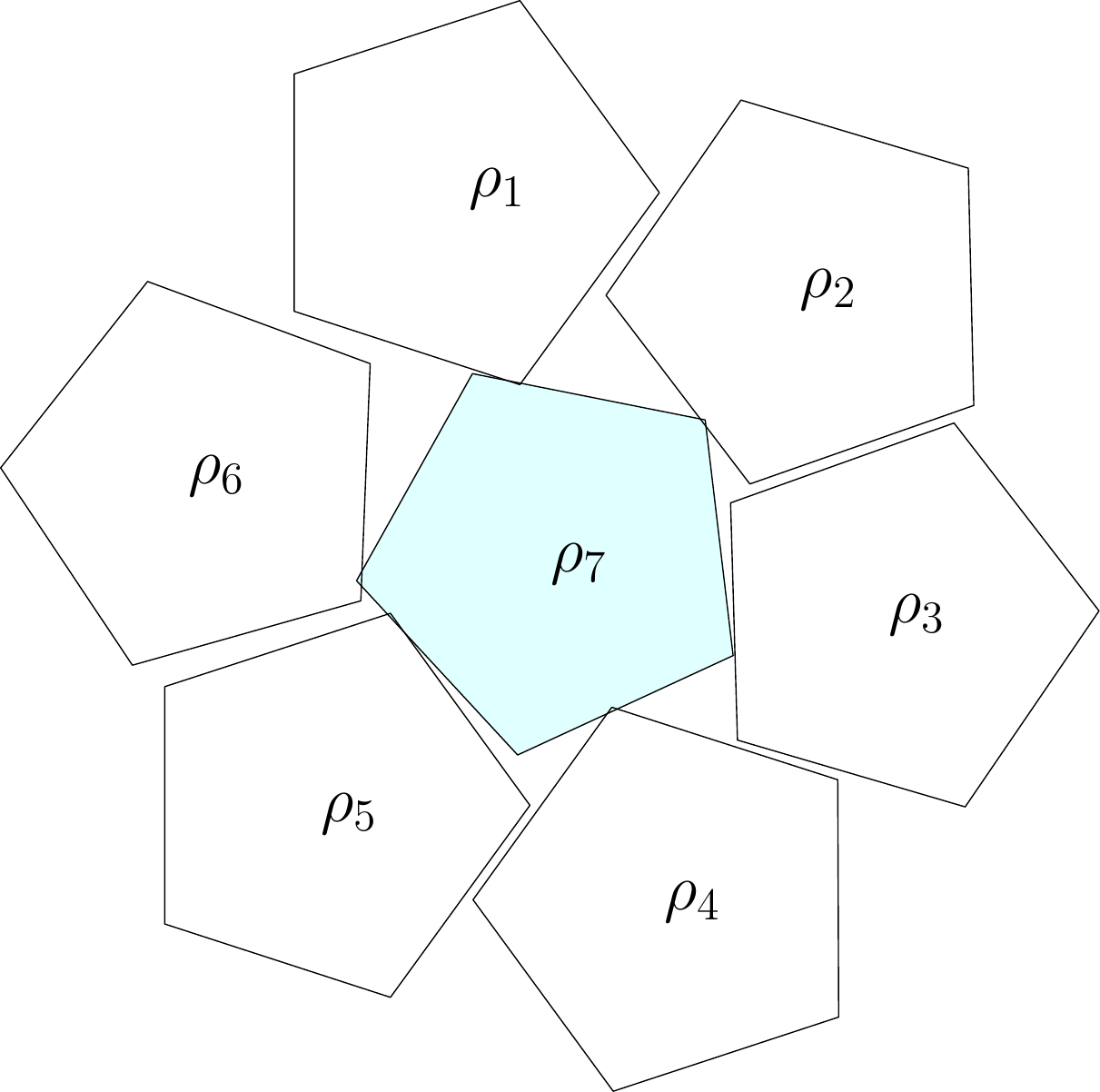} 
\caption{}
\label{fig:arb_pentagon}
     \end{subfigure}
     \hfill
     \begin{subfigure}[b]{0.32\textwidth}
         \centering
\includegraphics[scale=0.3]{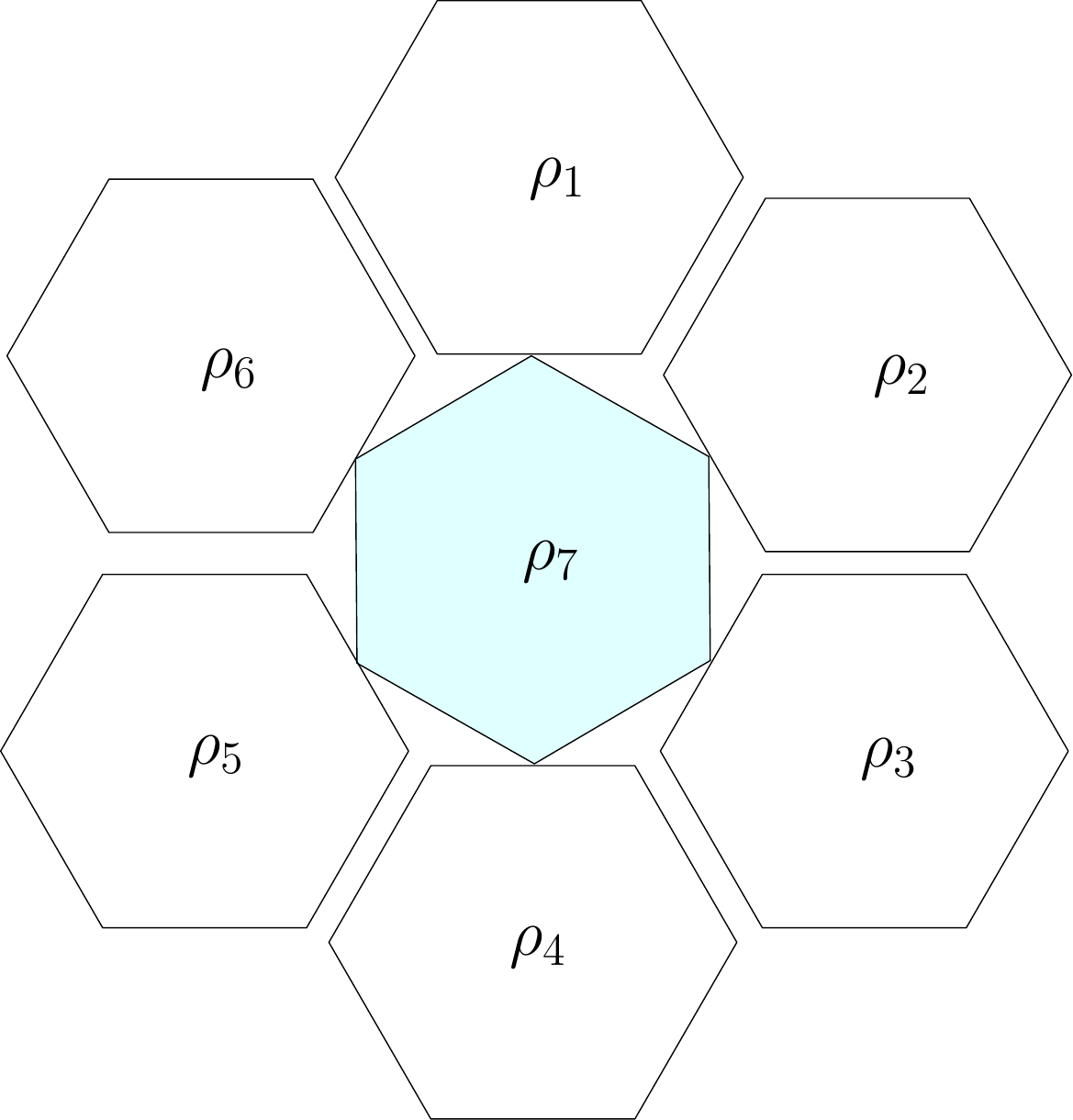}
\caption{}
\label{fig:arb_hexagon}
     \end{subfigure}
     \begin{subfigure}[b]{0.32\textwidth}
         \centering
\includegraphics[scale=0.3]{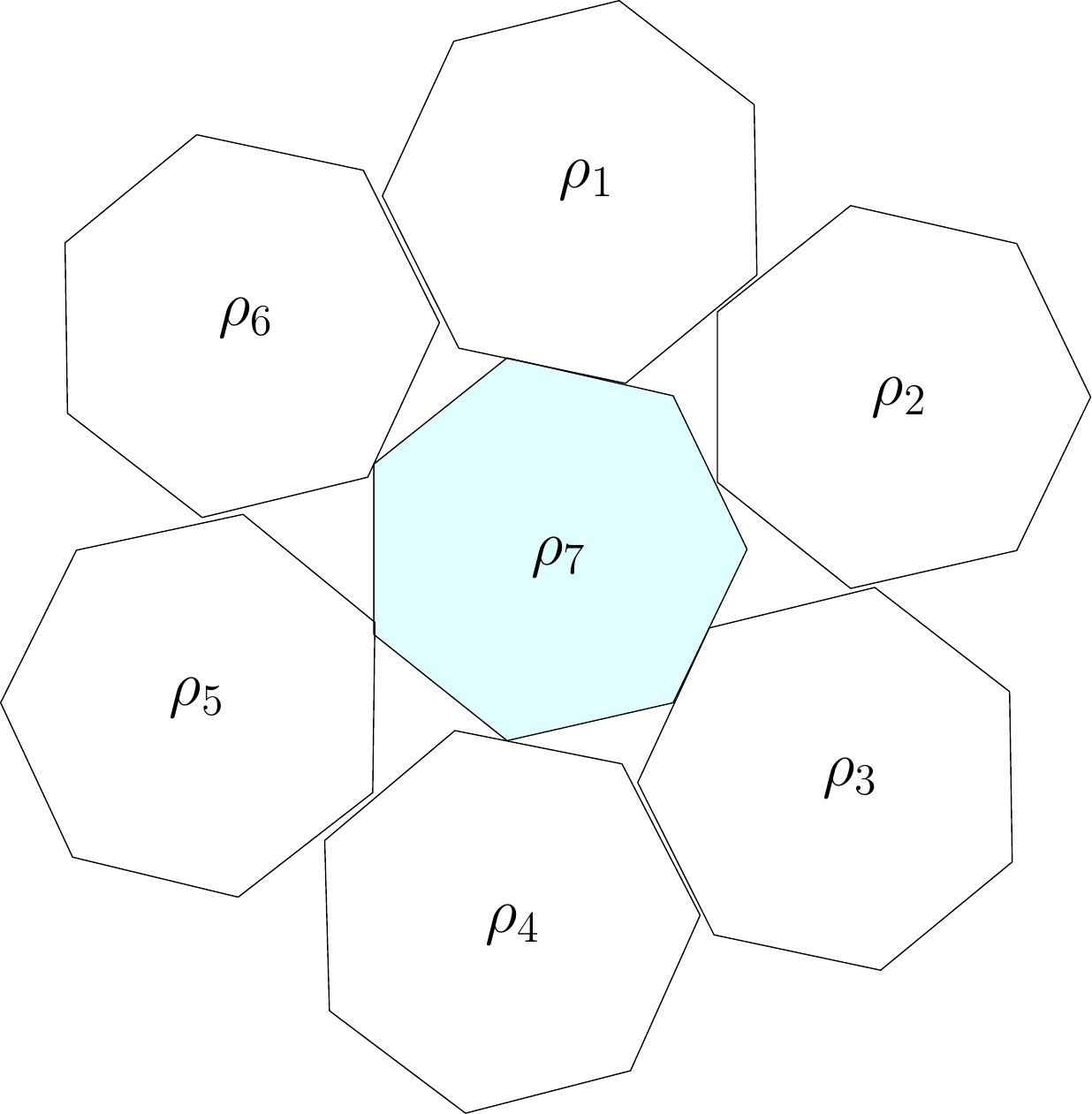}
\caption{}
\label{arb_septa}
     \end{subfigure}
     \begin{subfigure}[b]{0.32\textwidth}
         \centering
\includegraphics[scale=0.3]{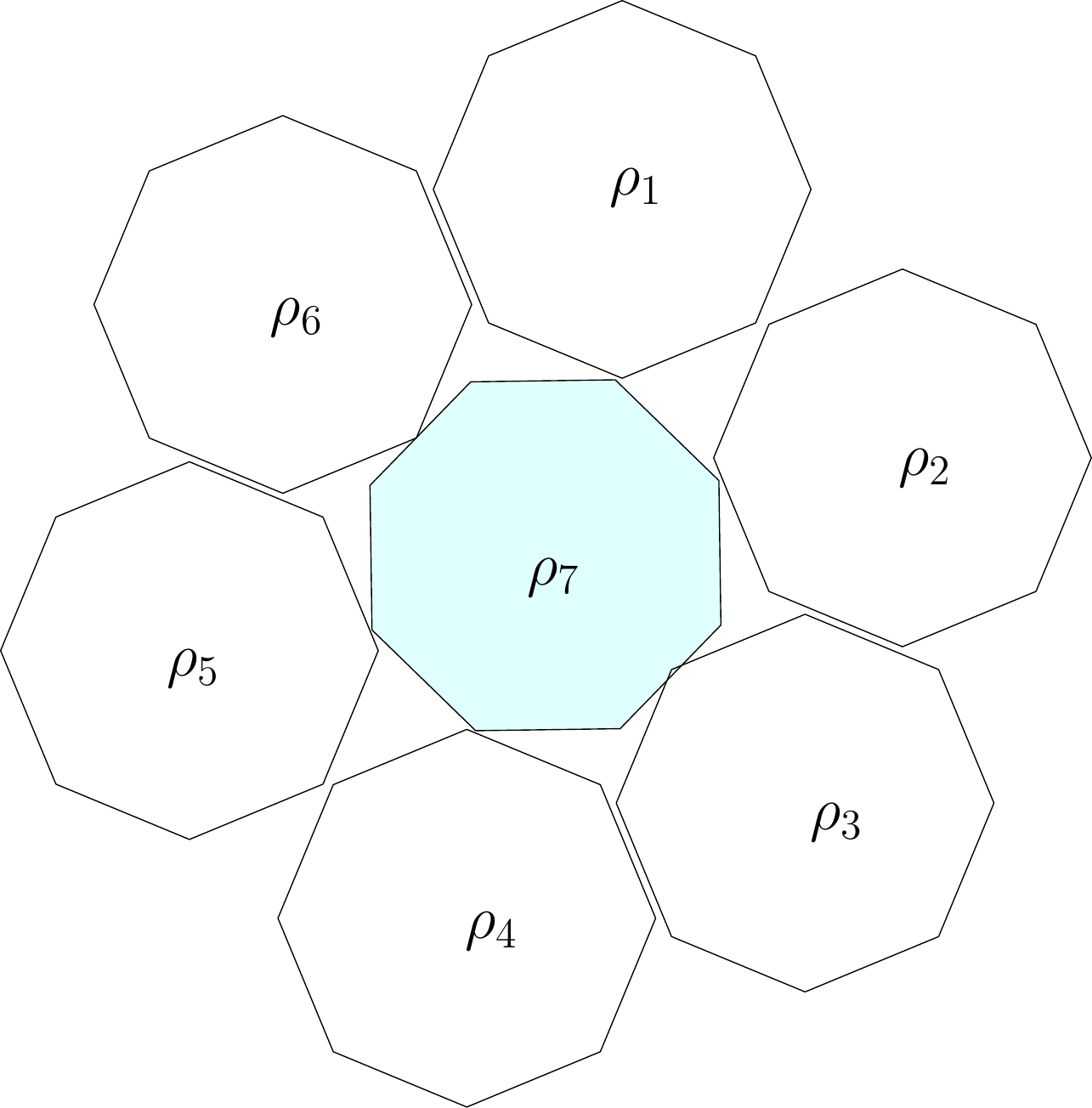}
\caption{}
\label{arb_octa}
\end{subfigure}
\begin{subfigure}[b]{0.32\textwidth}
\centering
\includegraphics[scale=0.3]{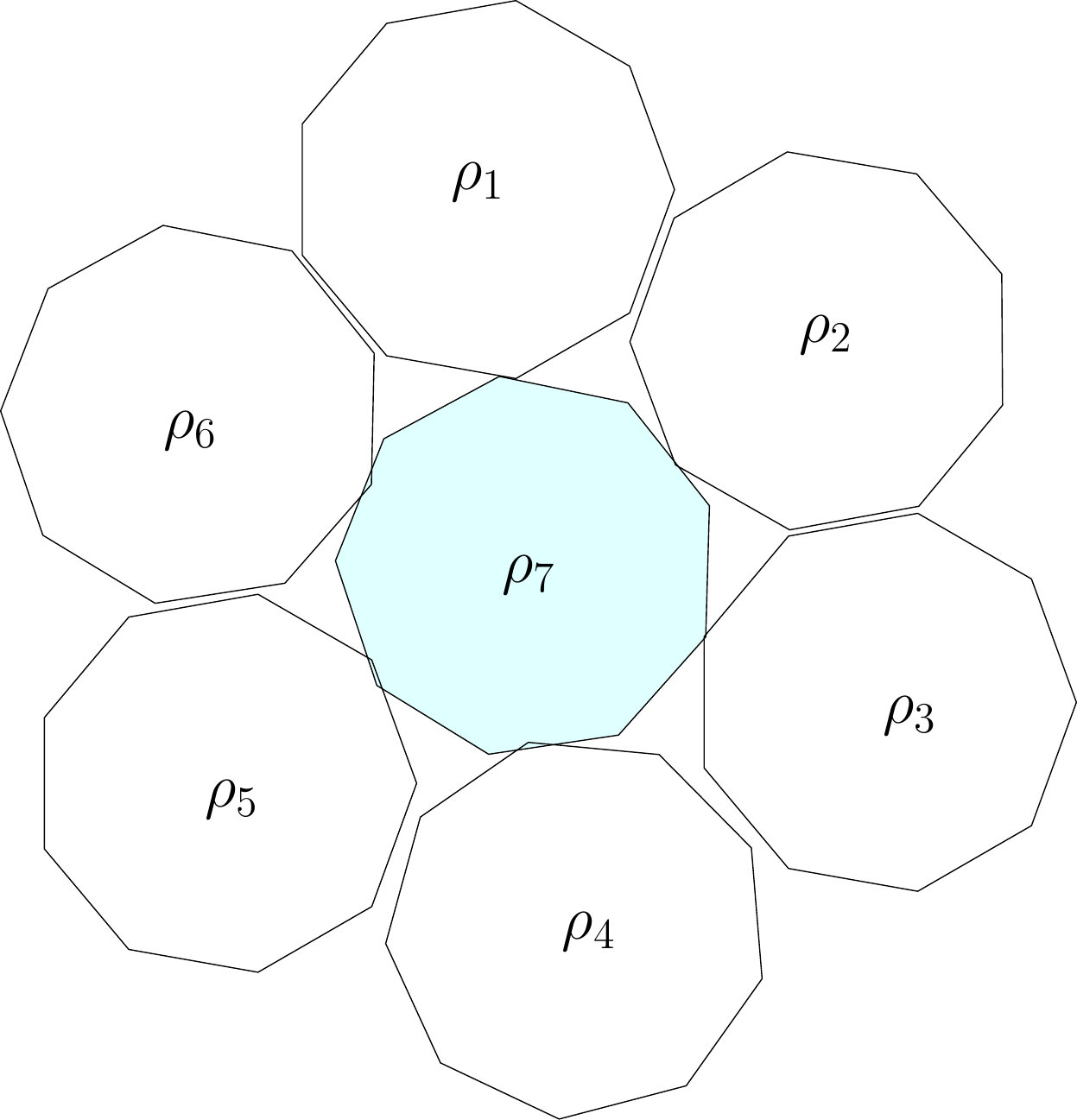}
\caption{}
\label{arb_nona}
\end{subfigure}
\begin{subfigure}[b]{0.32\textwidth}
\centering
\includegraphics[scale=0.3]{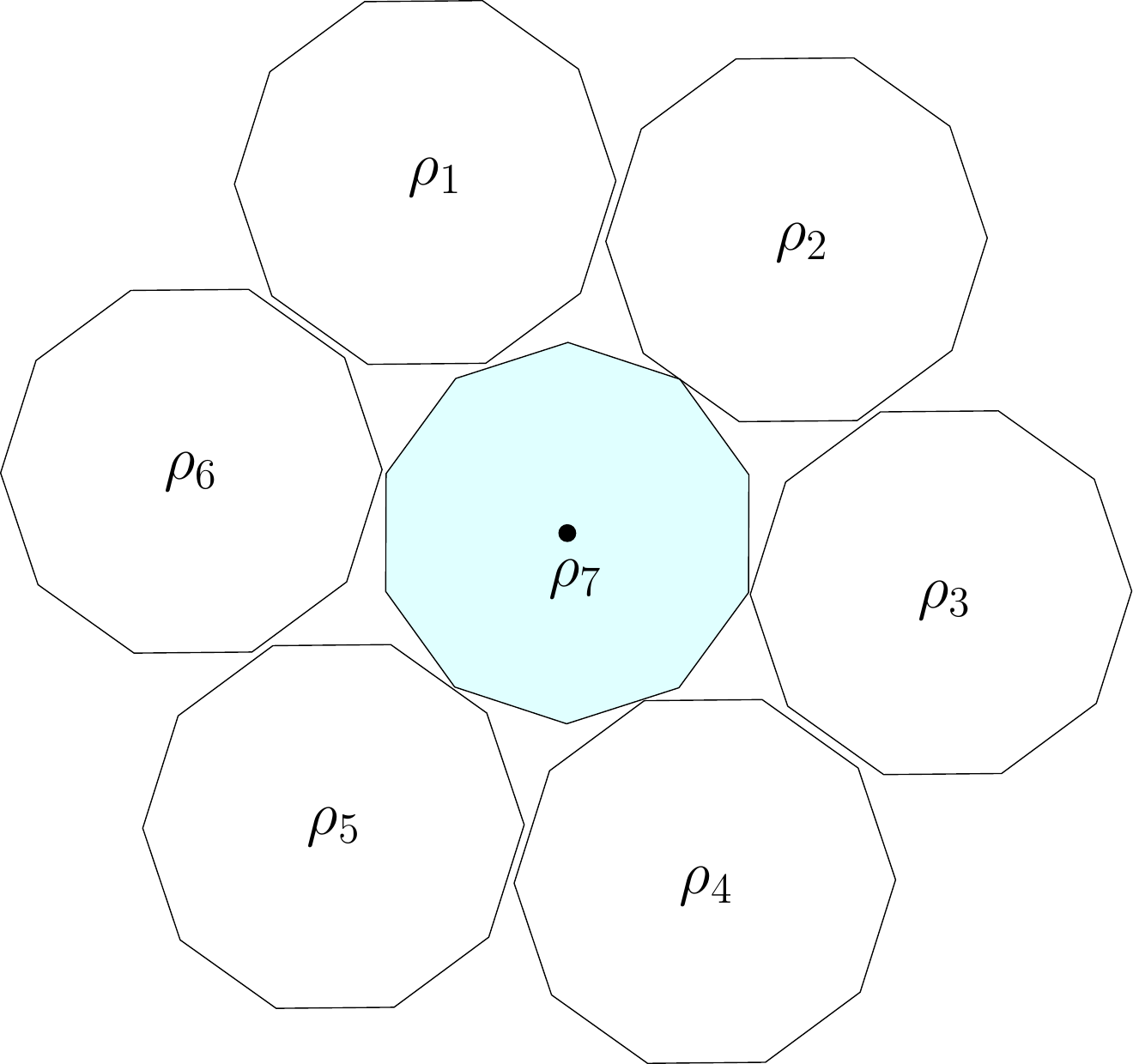}
\caption{}
\label{arb_deca}
\end{subfigure}
       \caption{ Independent kissing configuration for (a)~pentagon (b)~hexagon (c)~septagon (d)~octagon (e)~nonagon (f)~decagon.}
       \label{arb}
\end{figure}




\section{Conclusion}\label{10}
In this paper, we define independent kissing number $\zeta$ of a graph $G$ as $max_{v \in V}  \{\alpha(G[N(v)])\}$. The value of $\zeta$ for geometric objects like fixed oriented squares, unit disks and  fixed oriented $d$-dimensional hyper-cubes are 4, 5 and $2^d$, respectively; whereas the value of $\zeta$ for fixed oriented triangles, regular $k$-gons and arbitrary oriented squares belongs to $\{5,6\}, \{5,6\}$ and $\{6,7\}$, respectively. We show that  for  graphs  with  independent kissing number $\zeta$, the well known online greedy algorithm is optimal for minimum dominating-set, minimum independent dominating set and maximum independent set problems, and it achieves a competitive ratio of $\zeta$. For geometric intersection graph,  the well-known greedy algorithm $GCDS$ for MCDS achieves the asymptotic competitive ratio at most $2(\zeta-1)$; whereas the competitive ratio of every deterministic online algorithm  for MCDS for translates of a convex object  is at least $2(\zeta$-1), if $opt_{cds}$ is one; otherwise it is at least~$\frac{2(\zeta-1)+1}{3}$.
As an implication, it would be interesting combinatorial problem to know for which other graph classes the independent kissing number is constant.

\begin{figure}[h]
  \centering
     \begin{subfigure}[b]{0.45\textwidth}
         \centering
\includegraphics[scale=0.75]{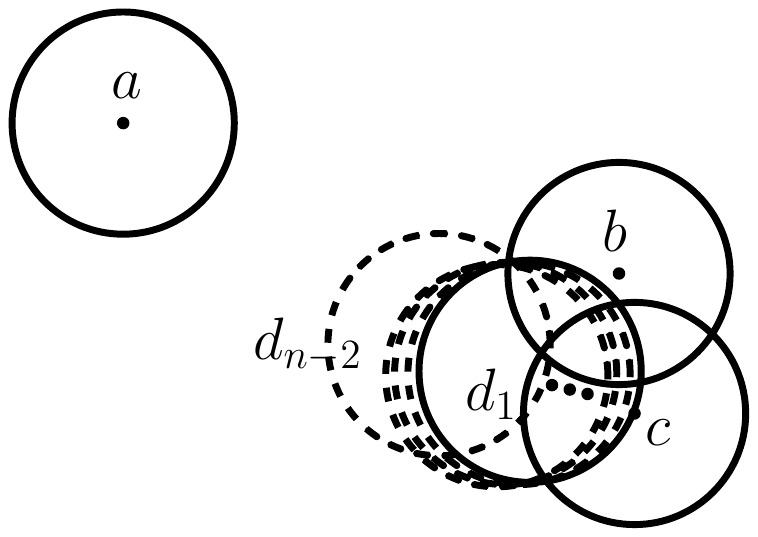} 
\caption{}
\label{fig:}
     \end{subfigure}
     \hfill
     \begin{subfigure}[b]{0.45\textwidth}
         \centering
\includegraphics[scale=0.75]{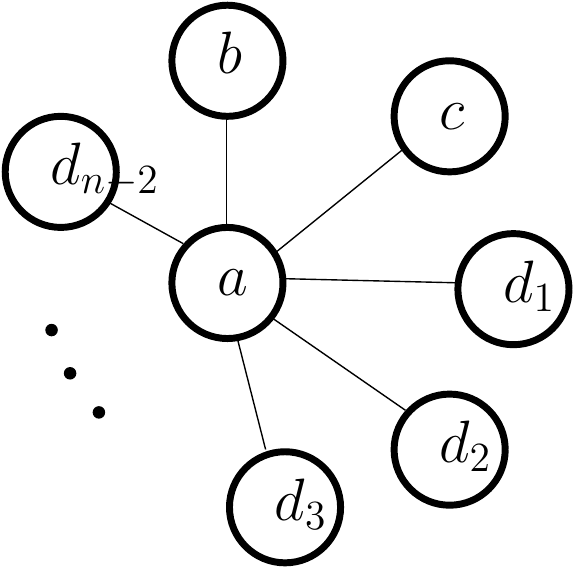}
\caption{}
\label{fig}
     \end{subfigure}
      \caption{(a) Set of unit disks, and (b)the  corresponding complement graph.}
       \label{fig:comparision}
\end{figure}

 Competitive ratio preserving reduction from  the maximum independent set  to the maximum clique for general graphs is known~\cite[Thm 7.3.3]{BorodinP}. As a result,  the lower bound for maximum clique is at least the lower bound of maximum independent set problem for general graphs. But  we cannot apply  this simple approach to graphs having  independent kissing number $\zeta$ because the independent kissing number may not be bounded for the complement graph $G_C$ of a graph $G$. For example, one can observe that $\zeta$ is unbounded (i.e.,  equal to the number of vertices $n$) for the complement graph of an unit disk graph (see Figure~\ref{fig:comparision}).  Whether one can obtain better result for the online maximum clique problem for geometric intersection graphs  remains as an open question.

 \subsection*{Acknowledgement}
The authors wish to thank  an anonymous reviewer  for referring to the article~\cite{1512873} 
that helps to improve the state of the art upper bound of the MCDS problem for unit disk graph in the online setup.





%
\bibliographystyle{siam}
\bibliography{ref}
\end{document}